\documentclass{article}
\usepackage{amsmath,amsxtra,amssymb,amsthm,amsfonts}
\usepackage{graphicx} 
\usepackage{color,graphicx}
\usepackage[margin=1.05in]{geometry}
\usepackage[pdftex,linktocpage=true,colorlinks,citecolor=blue,linkcolor=blue,pagebackref]{hyperref}
\usepackage{hyperref}
\usepackage{cleveref}
\usepackage[alphabetic,backrefs]{amsrefs}
\usepackage{algorithm}
\usepackage{algpseudocode}

\numberwithin{equation}{section}
\newtheorem{lemma}{Lemma}[section]
\newtheorem{prop}[lemma]{Proposition}
\newtheorem{theorem}[lemma]{Theorem}

\newtheorem{ex}[lemma]{Example}
\newtheorem{observation}[lemma]{Observation}

\newcommand{\zero}{\mathbf 0}
\newcommand{\uno}{\mathbf 1}
\newcommand{\be}{\mathbf e}

\newcommand{\bx}{\mathbf x}
\newcommand{\bu}{\mathbf u}
\newcommand{\bv}{\mathbf v}

\newcommand{\by}{\mathbf y}
\newcommand{\bw}{\mathbf w}

\newcommand{\bmu}{{\boldsymbol{\mu}}}

\newcommand{\supp}{{\Phi}}

\def\C{\mathbb{C}}
\def\R{\mathbb{R}}
\def\Z{\mathbb{Z}}

\def\ii{\mathrm{i}}

\usepackage[hang,flushmargin]{footmisc}

\begin{document}

\title{Perfect $(s,r)$-state transfer}

\author{ 
	Stephen Kirkland,\textsuperscript{1} \ Hermie Monterde,\textsuperscript{1,2} and  Sarah Plosker\textsuperscript{1,3}
}

\maketitle

\begin{abstract}
Much work has been done in the last two decades on the topic of quantum state transfer in a quantum spin network. One can model such a system of interacting qubits using an undirected graph, and studying vertex-to-vertex dynamics. This setup has recently been relaxed to allow for dynamics between linear combinations of two vertex states, i.e.\ from $\bu = \be_a + s \be_b$ to $\bmu=\be_{\alpha} + r \be_{\beta}$, where $r=s$ is either $-1$ (which corresponds to pair state transfer) or $+1$ (which corresponds to plus state transfer), or more recently $r=s$ is taken to be any real number (which corresponds to $s$-pair state transfer). 

Here, we broaden the investigation of $s$-pair state transfer to \textit{perfect $(s,r)$-state transfer}, which is perfect state transfer from $\bu = \be_a + s \be_b$ to $\bmu=\be_{\alpha} + r \be_{\beta}$ (up to some dilation) where $r,s\in \C$. We identify infinite families of graphs with perfect $(s,r)$-state transfer and provide characterizations of cases when $|r|= |s|$ and when $|r|\neq |s|$, showing situations when the degree of entanglement between vertex states is preserved, and when it is not preserved. The latter is particularly important as it represents perfect state transfer from an entangled pair of qubits to another one where the degree of entanglement need not be the same\mdash in fact, it can be set up so as to ``boost'' (increase) entanglement. We provide an algorithm that finds the vector with two nonzero entries that maximizes the fidelity of transfer for a fixed time $t$ starting from a given $s$-pair state $\bu$. Finally, we provide a sensitivity analysis, with respect to readout time errors, of perfect $(s,r)$-state transfer.

\medskip

\noindent \textbf{Keywords:} quantum walk, perfect state transfer, $s$-pair state, adjacency matrix, Laplacian matrix\\
	
\noindent \textbf{MSC2010 Classification:} 
05C50; 
81P45; 
05C76; 
15A18; 
81Q10 
 
\end{abstract}

\addtocounter{footnote}{1}
\footnotetext{Department of Mathematics, University of Manitoba, Winnipeg, MB, Canada R3T 2N2}
\addtocounter{footnote}{1}
\footnotetext{Department of Mathematics \& Statistics, University of Regina, Regina, SK, Canada  S4S 0A2
}
\addtocounter{footnote}{1}
\footnotetext{Department of Mathematics \& Computer Science, Brandon University, Brandon, MB, Canada R7A 6A9} 

{\def\thefootnote{}\footnotetext{Email: {\tt Stephen.Kirkland@umanitoba.ca, Hermie.Monterde@uregina.ca, ploskers@brandonu.ca}}

\section{Introduction}
The accurate transfer of quantum data, including quantum entanglement, between locations is a fundamental task in quantum information systems. 
A continuous-time quantum walk on a graph $X$ is described by a time-dependent transition operator $U(t)=\exp(itH)$, where $H$ is the  Hamiltonian describing the system whereby a 
quantum state evolves through the vertices of a graph. See \cite{kadian2021quantum} for a systematic review of quantum walks and their applications. Quantum walks through decision trees were initially considered in \cite{Farhi1998}, and this prompted a new class of quantum algorithms. By exploiting the interference
effects of quantum mechanics,  quantum walks can  outperform their classical counterparts (see \cite{childs2003, childs2007, childs2009},  and the references therein). On the other hand, transport
problems in quantum spin networks, specifically quantum spin chains, and the
fidelity of quantum state transfer, were considered in \cite{Bose:Quantum}. In this setting, the vertices of the graph are spins and the edges of the graph connect
spins which interact. The graph-theoretic distance between two vertices serves as a proxy for the physical distance between the locations of the two quantum spins. Over the past two decades, much work has been done on vertex state transfer (see \cite{christandl2004, christandl2005, godsil2012can, godsil2012state, kay2010perfect}). 

More recently, quantum state transfer between certain linear combinations of two vectors corresponding to vertices has been considered \cite{Chen2020PairST}: perfect \emph{pair state transfer} between  states of the form $\bu=\be_a- \be_b$ and $\bmu = \be_{\alpha} -  \be_{\beta}$ as well as perfect \emph{plus state transfer} between states of the form $\bu=\be_a+ \be_b$ and $\bmu = \be_{\alpha} +  \be_{\beta}$, both where the Hamiltonian is taken to the  Laplacian matrix of the graph of interest. For connected graphs on $5\leq n\leq 8$ vertices, it was shown that perfect pair state transfer occurs more often than perfect vertex state transfer for each such $n$, providing motivation for the study of these states. Note that the vertices involved in this type of transfer need not be adjacent, although they can be \cite{chen2019edge}. Perfect pair state transfer with respect to the adjacency matrix has been considered in \cite{Pal2024}. More generally, PST between states $\bu=\be_a+s \be_b$ and $\bmu = \be_{\alpha} + s \be_{\beta}$ when $s$ is real, called \emph{$s$-pair state transfer}, was explored in \cite{kim2024generalization} and recently in \cite{monterde2025perfect}.

There is some very recent literature on the subject of pure state transfer \cite{godsil2025perfect}, that performs spectral analysis on general real-valued unit vectors to characterize PST for these types of quantum states, and group state transfer \cite{brown2023continuous}, whereby the quantum spin network exhibits group state transfer if there exists two subsets of spins $S$ and $T$ for which submatrix of the time-dependent transition matrix $U(t)$ obtained by restricting to
columns in $S$ and rows not in $T$ is the all-zero matrix. These settings are quite general, and in many practical situations (e.g., \cite{chapman2016experimental}) one is interested in implementing
a perfect state transfer protocol applied to a photonic qubit entangled with another qubit at
a different location. The situation of quantum state transfer between two entangled quantum states is precisely the situation of $s$-pair state transfer discussed herein.

In this work, we develop the theory of $s$-pair state transfer. The literature in this area assumes that the constant $s$ is the same in both of the states, and so the term \emph{$s$-pair state transfer} was appropriate. In our more general setting where we allow for states $\bu=\be_a+r \be_b$ and $\bmu = \be_{\alpha} + s \be_{\beta}$, i.e., PST from an $s$-pair state to an $r$-pair state where $r, s\in \mathbb C$ , we use the term ``perfect $(s,r)$-state transfer".
Given the infancy of the topic, part of our contribution is to provide illustrative examples, outline key questions,  and make connections to provide better insight into the topic. We generalize from the real to the  complex setting,  consider examples of perfect $s$-pair state transfer, and provide optimality and sensitivity results in the absence of perfect $s$-pair state transfer. 
Section~\ref{sec:back} covers the necessary mathematical background while Section~\ref{sec:sPST} characterizes PST involving $s$-pair states in complete graphs and complete bipartite graphs. Section~\ref{sec:FR} uses known examples of fractional revival to create new examples of $s$-pair state transfer. Section~\ref{sec:optim} is devoted to the situation where perfect $(s,r)$-state transfer is not possible, but one desires to maximize the quantum state transfer. We provide an algorithm that finds, for a given graph, $s$-pair state $\bu$ and readout time $t,$ the vector with two nonzero entries that maximizes the fidelity of transfer at time $t$ starting from $\bu.$ Section~\ref{sec:sens} undertakes a sensitivity analysis of the fidelity (with respect to the readout time) in the presence of perfect $(s,r)$-state transfer, while Section~\ref{sec:concl} includes concluding remarks. 

Our results on PST for $s$-pair states are of two types: PST from $\be_a + s\be_b$ to
 $\be_\alpha + r\be_\beta$ where $|s|=|r|$ and PST from $\be_a + s\be_b$ to
 $\be_\alpha + r\be_\beta$, where $|s|$ may not be the same as $|r|.$ In the former case the degree of entanglement is preserved, and we mark those examples with a $*$. In the latter case the degree of entanglement may not be preserved and we denote  those examples with a $\dagger$. 

\section{Mathematical background}\label{sec:back}

Unless explicitly stated, we assume that $X$ is a simple (loopless) connected undirected weighted graph on $n$ vertices. We say that $X$ is unweighted if all edges in $X$ have weight equal to one. Suppose $H$ is a Hermitian matrix associated with $X$, that is, $H_{u,v}\neq 0$ if and only if there is an edge between $u$ and $v$ in $X$. The {\sl continuous-time quantum walk on $X$ with Hamiltonian $H$} has \textit{transition matrix}
\begin{equation*}
U_H(t) :=e^{\ii t H}.
\end{equation*}
Note that $U_H(t)$ is unitary since $H$ is Hermitian. Moreover, if $H$ is real symmetric (as is the case for the majority of this paper), then $U_H(t)$ is symmetric. Let $\operatorname{spec}(H)=\{\lambda_1,\ldots,\lambda_d\}$ denote the set of distinct eigenvalues of $H$.  If the spectral decomposition $H$ is given by
\begin{equation*}
H=\sum_{j=0}^d\lambda_jE_j,
\end{equation*}
where $E_j$ is the orthogonal projection matrix onto the eigenspace associated with $\lambda_j$, then we may write
\begin{equation*}
U_H(t)=\sum_{j=0}^de^{\ii t\lambda_j}E_j.
\end{equation*}
In this paper, we focus solely on the single excitation subspace of the Hamiltonian $H$; this subspace is characterised by the excitation number---
the number of 1's in the basis states $|0\rangle,  \dots, |n-1\rangle$---being equal to one \cite{keele2022combating}. Although the state of the quantum system as a whole is a unit vector in $(\C^2)^{\otimes n}$, the single-excitation framework reduces our study to $\C^n$ (see, e.g., \cite{chan2019quantum}). We also note that some results in this paper apply specifically to $H$ being the adjacency matrix $A=A(X)$ of $X$, or the Laplacian matrix $L=L(X)$. If we do not specify the Hamiltonian in a statement by using $U(t)$ to denote the
transition matrix, then one may assume it is any real symmetric matrix associated with $X$. In this case, the transition matrix is a symmetric unitary matrix.

A {\sl pure state} on a graph $X$ is a $1$-dimensional subspace of $\C^{n}$. We represent a pure state by a unit vector $\bu$ spanning the $1$-dimensional subspace. Note that $\gamma\bu$ represents the same state as $\bu$, for any phase factor $\gamma$. A pure state $\bu$ is a {\sl real state} if all its entries are real. A {\sl vertex state} is the characteristic vector $\be_a$ of some vertex $a$ in $X$.
For a non-zero complex number $s$, an {\sl $s$-pair state} is a pure state of the form
\begin{equation*}
\frac{1}{\sqrt{1+\vert s\vert^2}}\left(\be_a+s \be_b\right)
\end{equation*}
where $a$ and $b$ are two vertices in $X$ that need not be adjacent. In this work, we sometimes keep with the convention in the literature \cite{Chen2020PairST} of dropping the normalization factor for convenience. So we may simply refer to an $s$-pair state as a pure state of the form $\be_a+s\be_b$.  A $(-1)$-pair state is called {\sl pair state} if $s=-1$, and a $(1)$-pair state is called {\sl plus state}. 
An $s$-pair state 
represents a pair of entangled qubits forming a state in the 1-excitation subspace $\mathbb C^n$ of the full $2^n$-dimensional system of $n$ spins \cite{kim2024generalization}. 

Throughout, we let $\bu$ and $\bmu$ be linearly independent unit vectors in $\mathbb C^n$. We say that \textit{perfect state transfer} (PST) occurs from a pure state $\bu$ to a pure state $\bmu$ at time $\tau>0$ if
\begin{equation}
\label{PST}
U(\tau)\bu=\gamma\bmu.
\end{equation}
for some $\gamma\in\C$ called \textit{phase factor}. Since $U(t)$ is unitary, $|\gamma|=1$. If $\bu$ and $\bmu$ above are vertex states, then we also say \textit{perfect vertex state transfer} occurs between vertices $u$ and $v$ in $X$, while if $\bu$ is an $s$-pair state and $\bmu$ is an $r$-pair state, then we also say \textit{perfect $(s,r)$-state transfer}  occurs from $\bu$ to $\bmu$ ($s=r$ corresponds to \textit{perfect $s$-pair state transfer} in the literature). In particular, $(-1)$-pair state transfer is also called \textit{pair state transfer}, while $(+1)$-pair state transfer is also called \textit{plus state transfer}. If (\ref{PST}) holds but $\bu$ and $\bmu$ are linearly dependent, then we say that $\bu$ is periodic at time $\tau$. If every vertex state in $X$ is periodic, then we say that $X$ is a \textit{periodic graph}, in which case $U(\tau)=\gamma I$ \cite{godsil2011per}.

Suppose PST occurs from $\bu=\alpha\be_a+\beta\be_b$ to $\bmu=\eta\be_c+\zeta\be_d$ where $|\alpha|^2+|\beta|^2=|\eta|^2+|\zeta|^2=1$ and $\alpha,\beta,\eta,\zeta\in\C\backslash\{0\}$. Set $s=\frac{\beta}{\alpha}$ and $r=\frac{\zeta}{\eta}$. Then we may write (\ref{PST}) as
\begin{eqnarray}
\label{eq:pstsr}
U(\tau)(\be_a+s\be_b)=\gamma\frac{\eta}{\alpha}(\be_c+r\be_d)
\end{eqnarray}
for some $\gamma\in\C$ with $|\gamma|=1$. Consequently, PST from a pure state with exactly two nonzero entries to another pure state with the same property may be recast as PST from an $s$-pair state to a dilation of an $r$-pair state, where the dilation factor is $|\eta/\alpha|$. 

Physically speaking, an $s$-pair state $\be_a+s\be_b$ represents a pair of entangled qubits \cite{kim2024generalization}. The parameter $s$ may be viewed as a measure of the degree of entanglement of the qubits $a$ and $b$ in the $s$-pair state represented by $\be_a+s\be_b$. Consequently, in (\ref{eq:pstsr}), we have $|\eta/\alpha|=1$ if and only if $|s|=|r|$, and so the dilation factor $|\eta/\alpha|$ is an indicator of whether the entanglement has been preserved whenever perfect $(s,r)$-state transfer occurs. Given this interpretation, it follows that perfect $(s,r)$-state transfer (that is,
PST from an $s$-pair state to an $r$-pair state) represents a mapping of entangled states to entangled states. In particular, PST from an $s$-pair state to an $r$-pair state with $|s|=|r|$ can be viewed as PST on entangled states that have the same degree of entanglement. Perfect $(s,r)$-state transfer where the entanglement increases or decreases depends on the relative closeness of $|r|$ and $|s|$ to 1, which represents maximal entanglement. To clarify this point, we make the following definition. 
For an $s$-pair state $\be_a+s\be_b$, let 
\begin{eqnarray*}
m(s)\ :=\ \min\big\{|s|, |s|^{-1}\big\}.
\end{eqnarray*}
Note that $m(s)$ measures the degree of entanglement between the qubits $a$ and $b$ in the $s$-pair state $\be_a+s\be_b$. Observe that $m(s) \in [0,1]$, with $m(s)=1$ if and only if  the states $a$ and $b$ are maximally entangled. Consequently, pair and plus states are maximally entangled. The \textit{Bell states}, which are pure states in $\C^4$ of the form $\frac{1}{\sqrt{2}}(\be_1\pm \be_4),\frac{1}{\sqrt{2}}(\be_2\pm \be_3)$, are examples of pair and plus states. Moreover, we note that if $m(s)$ is close to $0$, then either $|s|$ is very small, or $|s|$ is very large. In the former case the normalized version $\frac1{\sqrt{1+|s|^2}} \be_a + \frac{s}{\sqrt{1+|s|^2}}\be_b$ is close to $\be_a$, while in the latter case $\frac1{\sqrt{1+|s|^2}} \be_a + \frac{s}{\sqrt{1+|s|^2}}\be_b$ is close to $\be_b$.
If we use $m(s)$ as a proxy for the degree of entanglement, then we can say that for perfect $(s,r)$-state transfer, the degree of entanglement is increased provided that $m(s) < m(r)$. 

From the equation $U(\tau)\bu=\gamma\bmu$, if $\bu$ and $\bmu$ are real-valued, one has reversibility (symmetry) of PST in the sense that if there is PST from $\bu$ to $\bmu$ then there is PST from $\bmu$ to $\bu$ at the same time and the same phase factor \cite[Cor 5.3]{godsil2017real}. For certain complex pure states, PST is reversible, but not in general  (\cite{kim2024generalization} provide examples where perfect $s$-pair state transfer is equivalent to perfect (vertex) state transfer, and it is well-known that vertex PST is reversible). However, it is the case that we always get  either: if there is PST from $\bu$ to $\bmu$  (at time $\tau$) then there is PST from $\overline{\bmu}$ to $\overline{\bu}$ (at time $\tau$); or if there is PST from $\bmu$ to $\bu$ (at time $t$) then there is PST from $\bmu$ to $\bu$  at time $\mathbf{-\tau}$. Reversing the sign of $t$ maps $U$ to its complex conjugate transpose: $U(-t)=U^*(t)$, and therefore the inverse transformation is being applied to the adjoint states. Both $U$ and $U^*$ are unitary operators, so neither is more physically meaningful than the other. Furthermore, the natural laws of physics are time-symmetric; the dynamics of the system are driven by the Schr\"{o}dinger equation 
\(\ii\hbar \frac{d}{dt} \vert{}\psi(t)\rangle = H \vert{}\psi(t)\rangle\), and there is no objective way to distinguish if a quantum particle is moving forward or backward in time. In short, negative time is indeed physically meaningful. See \cite{hoover2012time, sachs1987physics} for a more fulsome  discussion of time reversal.

The complex setting is  distinctly different from the real setting. The following provides an example of a 4-vertex graph where there is perfect $(s,r)$-state transfer from $\bu$ to $\bmu$ at time $\tau>0$, but no perfect $(s,r)$-state transfer from $\bmu$ to $\bu$ at any positive time. 
\begin{ex}\label{ex:p4}$\dagger$ {\rm{
    Consider the adjacency matrix of the path $P_4$ with end vertices $1$ and $4$. 
    Let $\bu=
    \frac{1}{\sqrt{2}}(\be_1+\ii\be_4)$. Let $\tau= \frac{2\pi}{\sqrt{5}}$ and $\theta = \pi\left(\frac{5+\sqrt{5}}{5}\right),$ and set $r=\frac{\sin(\theta+\frac{\pi}{4})}{\cos(\theta+\frac{\pi}{4})}\ii$. A computation shows that 
\begin{equation}
\label{p4}
U(\tau)\bu=\frac{1}{\sqrt{2}} 
\big(
  \cos \theta - \sin \theta\big)\be_1+\ii\big(  \cos \theta + \sin \theta\big)\be_4= 
    \cos\bigg(\theta + \frac{\pi}{4}\bigg)\be_1+\ii\sin\bigg(\theta + \frac{\pi}{4}\bigg)\be_4
=\frac{1}{\sqrt{1+|r|^2}}(\be_1+r\be_4),
\end{equation}
so that there is perfect $(s,r)$-state transfer from $\bu$ to $\bmu:=\frac{1}{\sqrt{1+|r|^2}}(\be_1+r\be_4)$ at time $\tau$, where 
$\bu$ is maximally entangled, but $\bmu$ is not since $|r|\ne 1$. This is an example where the degree of entanglement of the state of the system has decreased under perfect $(s,r)$-state transfer

Suppose that for some $\tau'>0$ we have perfect $(s,r)$-state transfer from $\bmu$ back to $\bu$. Then we get $U(\tau+\tau')\bu = \gamma \bu$ for some unit complex number $\gamma$.  The eigenvalues of $A$ are $\frac{\pm 1 \pm \sqrt{5}}{2}$ and we denote them by $\lambda_1 > ... > \lambda_4$. Since $U(\tau+\tau')\bu = \gamma \bu$, we get that
$e^{\ii(\tau+\tau')(\lambda_1-\lambda_j)}=1$ for $j=2,3,4$. Taking $j=2$ yields $\tau+\tau' = 2 \pi k$ for some positive integer $k$, and $j=3$ yields  $(\tau+\tau')\sqrt{5} = 2 \pi m$ for some positive integer $m$, a contradiction. Thus no such $\tau'>0$ exists. }}

{\rm{We note in passing that, by taking complex conjugates in (\ref{p4}), it is readily established that $U(\tau) \overline{\bmu}=\overline{\bu}.$ Note that $\overline{\bu}$ is still maximally entangled, whereas $\overline{\bmu}$ is not. This thus yields an example where the degree of entanglement 
of the state of the system has increased under perfect $(s,r)$-state transfer. 
 }}
\end{ex}

\begin{ex}\label{ex:p4'} $\dagger$ \rm{
Let us revisit Example \ref{ex:p4}. For $k \in \mathbb{N},$ let $\tau_k= \frac{2\pi k}{\sqrt{5}}$ and $\theta_k = \pi k\left(\frac{5+\sqrt{5}}{5}\right). $ Note that $$U(\tau_k) \bu = 
    \cos\bigg(\theta_k + \frac{\pi}{4}\bigg)\be_1+\ii\sin\bigg(\theta_k + \frac{\pi}{4}\bigg)\be_4\equiv \bmu_k $$ for each such $k$, so that there is PST from $\bu$ to $\bmu_k$ at time $\tau_k$ for each $k \in \mathbb{N}.$ 

Fix a value $\alpha \in [0,\pi].$ Invoking Kronecker's approximation theorem, we find that for each $\epsilon >0,$ there  are $k, h \in \mathbb{N}$ such that $\Big|k\left(\frac{5+\sqrt{5}}{10}\right) -h + \frac{1}{8}-\frac{\alpha}{2\pi}\Big| < \frac{\epsilon}{2\pi},$ and hence $\Big| \theta_k+\frac{\pi}{4} - 2\pi h -\alpha \Big| < \epsilon.$ Consequently, there is a subsequence $k_j \in \mathbb{N}$ such that $\cos(\theta_{k_j}+\frac{\pi}{4}) \rightarrow \cos(\alpha)$ and  $\sin(\theta_{k_j}+\frac{\pi}{4}) \rightarrow \sin(\alpha)$ as $j \rightarrow \infty.$ Consequently,  $$\lim_{j \rightarrow \infty} \bmu_{k_j} =\cos(\alpha)\be_1+\ii\sin(\alpha)\be_4.$$
Hence, not only is there PST from $\bu$ to $\bmu_k$ at time $\tau_k$ for each $k \in \mathbb{N},$ but also, via a suitable choice of $k,$ the degree of entanglement in $\bmu_k$ can be made arbitrarily close to any prespecified degree of entanglement.}
\end{ex}

It should be noted that if the degree of entanglement decreases under perfect $(s,r)$-state transfer, one can reverse the time as described earlier to force the entanglement to increase. That is, so long as the degree of entanglement is not preserved, i.e., $m(s)\neq m(r)$, then we can reverse time so that the degree of entanglement is increased ($m(s) < m(r)$). This entanglement boosting effect may be of physical interest, as entanglement is a key quantum resource. 
\section{Perfect $(s, r)$-state transfer}\label{sec:sPST}

Let  $\bu$ and $\bmu$ be complex vectors. The eigenvalue support of $\bu$ (relative to $H$) is the set
\begin{equation*}
\supp_{\bu}=\{\lambda_j\in\operatorname{spec}(H):E_j\bu\neq \zero\}.
\end{equation*}
We say that $\bu$ is a \textit{fixed state} if $|\supp_{\bu}|=1$. It is known that $\bu$ is a fixed state if and only if it is an eigenvector for $H$. A fixed state is not involved in PST \cite{godsil2025perfect,monterde2025quantum}. We say that $\bu$ and $\bmu$ are strongly cospectral (relative to $H$) if for each $j$, there is a unit complex numbers $\gamma_j$ such that
\begin{equation*}
E_j\bu=\gamma_jE_j\bmu.
\end{equation*} 
Note that if $\bu$ and $\bmu$ are real vectors, then each $\gamma_j\in\{\pm 1\}$. Analogous to the vertex PST setting, we have the following observation (for proof, see \cite[Thm 4.1.3]{monterde2025quantum}). 

\begin{observation}
\label{Thm:charPST}
Let $\bu,\bmu\in \C^m$ and $H$ be a real symmetric matrix. PST occurs from $\bu$ to $\bmu$ relative to $H$ if and only if the following conditions hold.
\begin{enumerate}
\item The complex vectors $\bu$ and $\bmu$ are strongly cospectral. That is, for each $j$, we have $E_j\bu=\gamma_jE_j\bmu$ for some unit $\gamma_j\in\C$.
\item There exists $\tau>0$ such that for any two eigenvalues $\lambda_j,\lambda_h\in\supp_{\bu}=\supp_{\bmu}$, we have $\gamma_je^{\ii\tau\lambda_j}=\gamma_h e^{\ii\tau\lambda_h}.$
\end{enumerate}
\end{observation}

Our goal in this section is to characterize the existence of perfect $(s,r)$-state transfer in complete graphs and complete bipartite graphs relative to $A$ and $L$.


\subsection{Complete graphs}

The spectral decomposition of the transition matrix of $K_n$ relative to $A$ is given by 
\begin{equation}
\label{Eq:kn}
U_A(t)=e^{-\ii t}\left((e^{\ii tn}-1)\frac{1}{n}J+I\right).
\end{equation}

\begin{prop}
\label{Cor:knPST}
Let $V(K_n)=\{1,\ldots,n\}$. Suppose $\bu=\be_1+s\be_2$ with $s\in\C\backslash\{0,-1\}$ and $\bmu\in\C^n$ is a vector that has two nonzero entries. PST occurs from $\bu$ to $\bmu$ in $K_n$ relative to $A$ at time $\tau$ if and only one of the following conditions hold.
\begin{enumerate}
\item $\dagger$  $n=2$, $s\in\C\backslash\{0\}$, $\bmu=\frac{e^{\ii2\tau}+1+s(e^{\ii2\tau}-1)}{2}\be_1+\frac{e^{\ii2\tau}-1+s(e^{\ii2\tau}+1)}{2}\be_2$, and $\tau\in\R$.
\item $*$ $n=3$, $s=\frac{2+e^{\ii3\tau}}{1-e^{\ii 3\tau}}$, $\bmu=\be_3+(s-1)\be_2$, and $\tau\in\R$.
\item $*$ $n=3$, $s=\frac{1-e^{\ii 3\tau}}{2+e^{\ii 3\tau}}$, $\bmu=s\be_3+(1-s)\be_1$, and $\tau\in\R$.
\item $*$ $n=4$, $s=1$, $\bmu=\be_3+\be_4$ and $\tau=\frac{\pi}{4}$.
\end{enumerate}
\end{prop}

\begin{proof}
Let $\bu=\be_1+s\be_2$, where $s\in\C\backslash\{0\}$. If $\bu$ admits PST in $K_n$, then it must be that $s\neq \pm 1$, otherwise $\bu$ is a fixed state. From (\ref{Eq:kn}), we get PST from $\bu$ to $\bmu$ at time $\tau$ with phase factor $e^{-\ii\tau}$ if and only if 
\begin{center}
$\bmu=(\be_1+s\be_2)+\frac{1}{n}(e^{\ii n \tau}-1)(1+s)\uno.$
\end{center}
Note that $\bmu$ has at least three nonzero entries whenever $n\geq 5$. Thus, it must be $n\in\{2,3,4\}$. Moreover, we have that $e^{\ii n \tau}\neq 1$, otherwise $\bu=\bmu$ is periodic at $\tau$. If $n=2$, then $\bmu=\frac{1}{2}\left[\begin{array}{cc} e^{\ii2\tau}+1+s(e^{\ii2\tau}-1) \\ e^{\ii2\tau}-1+s(e^{\ii2\tau}+1) \end{array} \right]$. If $n=3$, then PST occurs from $\bu$ to $\bmu=\bu+\frac{1}{3}(e^{\ii3\tau }-1)(1+s)\uno$ at time $\tau$.  Now, $\bmu$ has two nonzero entries if and only if either $1+\frac{1}{3}(e^{\ii3\tau }-1)(1+s)=0$ or $s+\frac{1}{3}(e^{\ii3\tau }-1)(1+s)=0$. Equivalently, $s=\frac{2+e^{\ii3\tau}}{1-e^{\ii 3\tau}}$ or $s=\frac{1-e^{\ii 3\tau}}{2+e^{\ii 3\tau}}$. In particular, if $s=\frac{2+e^{\ii3\tau}}{1-e^{\ii 3\tau}}$, then PST occurs from $\bu$ to $\bmu=\be_3+(1-s)\be_2$. Meanwhile, if $s=\frac{1-e^{\ii 3\tau}}{2+e^{\ii 3\tau}}$, PST occurs from $\bu$ to $\bmu=s\be_3+(s-1)\be_1$. Finally, if $n=4$, then $\bmu$ has two nonzero entries if and only if $1+\frac{1}{4}(e^{\ii4\tau }-1)(1+s)=0$ and $s+\frac{1}{4}(e^{\ii4\tau }-1)(1+s)=0$. That is, $s=\frac{3+e^{\ii 4\tau}}{1-e^{\ii 4\tau}}=\frac{1-e^{\ii 4\tau}}{3+e^{\ii 4\tau}}$. Equivalently, $\tau=\frac{\pi}{4}$ and $s=1$. In this case, PST occurs between $\be_1+\be_2$ and $\be_3+\be_4$ in $K_4$ at $\frac{\pi}{4}$.
\end{proof}

In Proposition \ref{Cor:knPST} 1, there are choices of $\tau$ and $s$ so that the degrees of entanglement in $\bu$ and $\bmu$ are different. In Proposition \ref{Cor:knPST} 2, $s-1=-\overline{s}$ so that the degree of entanglement is preserved in the target $s$-pair state. Similarly in Proposition \ref{Cor:knPST} 3, $|1-s|=1$, so the degree of entanglement is again preserved. It is also immediate from Proposition \ref{Cor:knPST} that if $n\geq 5$, then perfect $(s,r)$-state transfer does not occur in $K_n$. This motivates us to study optimization techniques in Section~\ref{sec:optim}.


\subsection{Complete bipartite graphs}

For complete bipartite graphs, we establish results analogous to Corollaries \ref{Cor:knPST} relative to $A$ and $L$. Since $K_2$ is dealt with, we only examine $K_{a,b}$ with $a+b\geq 3$. We start with the case when $H=A$.

\subsubsection{Adjacency case}

The spectral decomposition of the transition matrix of $K_{a,b}$ relative to $A$ is given by
\begin{equation}
\label{Eq:completebip}
U_A(t)=\begin{bmatrix}
I_a-\frac{1}{a}J&0 \\ 0&I_b-\frac{1}{b}J
\end{bmatrix}+\frac{e^{\ii t\sqrt{ab}}}{2ab}\begin{bmatrix}
bJ&\sqrt{ab}J \\ \sqrt{ab}J&aJ
\end{bmatrix}+\frac{e^{-it\sqrt{ab}}}{2ab}\begin{bmatrix}
bJ&-\sqrt{ab}J \\ -\sqrt{ab}J&aJ
\end{bmatrix}.
\end{equation}

\begin{theorem}
\label{Thm:kmnPSTc1}
Let $V(K_{a,b})=B_1\cup B_2$, where $|B_1|=a$. Suppose $\bu=\be_u+s\be_v$ for some $s\in\C\backslash\{0\}$ and $\bmu\in\C^n$ is a vector with two nonzero entries. PST occurs from $\bu$ to $\bmu$ in $K_{a,b}$ relative to $A$ at time $\tau$ if and only if one of the following conditions hold.
\begin{enumerate}
\item $u,v\in B_1$, and one of the conditions below holds.
\begin{enumerate}

\item $*$ $a=2$, $b\geq 1$, $s\neq 1$, $\bmu=-(\be_v+s\be_u)$ and $\tau=\frac{\pi}{\sqrt{2b}}$. 

\item $*$ $a=4$, $b\geq 1$, $s=1$, $\bmu=\be_w+\be_x$ where $B_1=\{u,w,v,x\}$, and $\tau=\frac{\pi}{2\sqrt{b}}$. 

\item $*$ $a=3$, $b\geq 1$, $s=\frac{1}{2}$, $\bmu=\be_w+\frac{1}{2}\be_x$ where $B_1=\{u,w,v\}$, and $\tau=\frac{\pi}{\sqrt{3b}}$. 

\item $\dagger$ $a=2$, $b=1$, $s=\cot^2(\tau/\sqrt{2})
$, $\bmu=(s-1)\be_v+\frac{1+s}{\sqrt{2}}\ii\sin(\tau\sqrt{2})\be_w$ where $B_2=\{w\}$, and $\tau\in\R$.
\item $\dagger$ $a=2$, $b=1$, $s=\tan^2(\tau/\sqrt{2})$, $\bmu=(1-s)\be_u+\frac{1+s}{\sqrt{2}}\ii\sin(\tau\sqrt{2})\be_w$ where $B_2=\{w\}$, and $\tau\in\R$.

\item $*$ $a=3$, $b=1$, $s=1$, $(\bmu,\tau)\in\{\big(-\be_x+\ii\be_w,\frac{2\pi(3k+1)}{3\sqrt{3}}\big),\big(-\be_x-\ii\be_w,\frac{2\pi(3k+2)}{3\sqrt{3}}\big)\}$ where $B_1=\{u,v,x\}$, $B_2=\{w\}$ and $k\in\Z$.

\item $*$ $a=2$, $b=2$, $s=1$, and $\bmu=\ii(\be_w+\be_x)$ where $B_2=\{w,x\}$, and $\tau=\frac{\pi}{4}$.
\end{enumerate}

\item $u\in B_1$, $v\in B_2$ and one of the conditions below hold.
\begin{enumerate}
\item $*$ $a=1$, $b=2$, $s\in\C\backslash\{0\}$, $\bmu=\be_u+s\be_w$ where $B_2=\{v,w\}$, and $\tau=\frac{\pi}{\sqrt{2}}$.
\item $*$ $a=b=2$, $s\in\C\backslash\{0\}$, $\bmu=\be_w+s\be_x$ where $B_1=\{u,w\}$ and $B_2=\{v,x\}$, and $\tau=\frac{\pi}{2}$.
\item $*$ $a=1$, $b\geq 1$, $s=\ii\sqrt{b}\cot(\tau\sqrt{b}/2)
$, $\bmu=-\be_u+s\be_v$, and $\tau\in\R$. 
\item $\dagger$ $a=1$, $b=2$, $s=-\bigg(\frac{1}{\sqrt{2}\csc(\tau\sqrt{2})-\sqrt{1/2}\tan(\tau/\sqrt{2})}\bigg)\ii $, $\bmu=\left(\cos( \tau\sqrt{2})+\ii \frac{s}{\sqrt{2}}\sin(\tau\sqrt{2})\right)\be_u -s\be_w$
where $B_2=\{v,w\}$, and $\tau\in\R$.
\item $\dagger$ $a=1$, $b=2$, $s=-\sqrt{2}\cot(\tau\sqrt{2})$, 
$B_2=\{v,w\}$, $\bmu=\left(\ii\frac{1}{\sqrt{2}}\sin(\tau\sqrt{2})+\frac{s}{2}(\cos(\tau\sqrt{2})+1)\right)\be_v+\left((\ii\sin(\tau\sqrt{2})\frac{1}{\sqrt{2}}+\frac{s}{2}(\cos(\tau\sqrt{2})-1)\right)\be_w$, and $\tau\in\R$.
\item $*$ $a=1$, $b=3$, $(s,\bmu,\tau)\in\left\{\big(1,\ii(\be_w+\be_x),\frac{2\pi(3k+1)}{3\sqrt{3}}\big),\big(1,-\ii(\be_w+\be_x),\frac{2\pi(3k+2)}{3\sqrt{3}}\big)\right\}$, where $B_2=\{v,w,x\}$ and $k\in \Z$.

\item $*$ $a=2$, $b\geq 1$, $s=\ii \sqrt{b/2} \cot(\tau\sqrt{b/2})$, $\bmu=-\be_w+s\be_v$ where $B_1=\{u,w\}$, and $\tau\in\R$.
\end{enumerate}
\end{enumerate}
Moreover, the above statements all hold if we interchange the roles of $a$ and $b$, as well as $B_1$ and $B_2$.
\end{theorem}

\begin{proof}
First, let $u,v\in B_1$ so that $a\geq 2$. Then $s\neq -1$, otherwise $\bu$ is a fixed state. Letting $\bu_1=\be_u+s\be_v\in\C^a$ and $\bu_2=\zero_b$, (\ref{Eq:completebip}) gives us
\begin{equation*}
\bmu=U_A(\tau)\begin{bmatrix}
\bu_1\\ \bu_2\end{bmatrix}=\begin{bmatrix}
\be_u+s\be_v+\frac{1+s}{a}\big(\cos( \tau\sqrt{ab})-1\big)\uno_a \\ \frac{1+s}{\sqrt{ab}}\ii\sin( \tau\sqrt{ab})\uno_b
\end{bmatrix}. 
\end{equation*}
We have two cases. 

\noindent \textbf{Case 1.} Suppose $\sin( \tau\sqrt{ab})=0$. That is, $\tau=\frac{\ell\pi}{\sqrt{ab}}$ for any integer $\ell$. If $\ell$ is even, then $\bu=\bmu$ is periodic at $\tau$. If $\ell$ is odd, 
$\bmu$ has two nonzero entries if and only if $a=2$ and $s\neq -1$, or $(s,a)\in\{(1,4),(\frac{1}{2},3)\}$. If $a=2$ and $s\neq -1$, then $\bmu=-(\be_v+s\be_u)$, and so 1a holds. If $(s,a)=(1,4)$, then PST occurs between $\be_u+\be_v$ and $\be_w+\be_x$ at $\frac{\pi}{2\sqrt{b}}$, where $B_1=\{u,w,v,x\}$. So 1b holds. If $(s,a)=(\frac{1}{2},3)$, then PST occurs between $\be_u+\frac{1}{2}\be_v$ and $\be_w+\frac{1}{2}\be_v$ at $\frac{\pi}{\sqrt{3b}}$, where $B_1=\{u,w,v\}$. This proves 1c. 

\noindent \textbf{Case 2.} Suppose that $\sin( \tau\sqrt{ab})\neq 0$. Since $s\neq -1$ and $\bmu$ has exactly two nonzero entries, we have $b\leq 2$. If $b=1$, then
\begin{equation*}
\bmu=\begin{bmatrix}
\be_u+s\be_v+\frac{1+s}{a}(\cos(\tau\sqrt{a})-1)\uno_a \\ \frac{1+s}{\sqrt{a}}\ii\sin(\tau\sqrt{a})
\end{bmatrix}
\end{equation*}
has two nonzero entries if and only if (i) $a=2$ and $\frac{1+s}{2}(\cos(\tau\sqrt{2})-1)\in\{-1,-s\}$ or (ii) $a=3$, $s=1$ and $\tau=\frac{2\pi(3k+\ell)}{3\sqrt{3}}$, where $\ell\in\{1,2\}$ and $k\in\Z$. Note that (i) yields the values 
\begin{equation*}
s=\dfrac{1+\cos(\tau\sqrt{2})}{1-\cos(\tau\sqrt{2})}=\cot^2(\tau/\sqrt{2})\quad  \text{and} \quad s=\dfrac{1-\cos(\tau\sqrt{2})}{1+\cos(\tau\sqrt{2})}=\tan^2(\tau/\sqrt{2})
\end{equation*}
and the vectors $\bmu=(s-1)\be_v+\frac{1+s}{\sqrt{2}}\ii\sin(\tau\sqrt{2})\be_w$ and $\bmu=(1-s)\be_u+\frac{1+s}{\sqrt{2}}\ii\sin(\tau\sqrt{2})\be_w$,
respectively, where $B_2=\{w\}$. This proves 1d and 1e. Meanwhile, (ii) gives us $\bmu=-\be_x+\ii\be_w$ whenever $\ell=1$ and $\bmu=-\be_x-\ii\be_w$ whenever $\ell=2$, where $B_1=\{u,v,x\}$ and $B_2=\{w\}$. This establishes 1f. On the other hand, if $b=2$, then $\bmu$ has exactly two nonzero entries if and only if $\be_u+s\be_v+\frac{1+s}{a}(\cos(\tau\sqrt{2a})-1)\uno=0$. Equivalently, $a=2$, $s=1$, $\tau=\frac{\pi}{4}$ and $\bmu=\ii (\be_w+\be_x)$, where $B_2=\{w,x\}$. Thus, 1g holds.

Next, suppose that $u\in B_1$ and $v\in B_2$. We again use (\ref{Eq:completebip}) to obtain an expression for $\bmu=U_A(t)\begin{bmatrix}
\bu_1\\ \bu_2\end{bmatrix}$, where $\bu_1=\be_u$ and $\bu_2=s\be_v$. One checks that
\begin{equation}
\label{Eq:pf}
\bmu=\begin{bmatrix}
\be_u+\left(\big(\cos( \tau\sqrt{ab})-1\big)\frac{1}{a}+\ii s\sin(\tau\sqrt{ab})\frac{1}{\sqrt{ab}}\right)\uno_a \\ s\be_v+\left(\ii\sin(\tau\sqrt{ab})\frac{1}{\sqrt{ab}}+s(\cos(\tau\sqrt{ab})-1)\frac{1}{b}\right)\uno_b
\end{bmatrix}. 
\end{equation}
We again have two cases.

\noindent \textbf{Case 1.} Assume $\sin( \tau\sqrt{ab})=0$. That is, $\tau=\frac{\ell\pi}{\sqrt{ab}}$ for any integer $\ell$. If $\ell$ is even, then $\bu=\bmu$ is again periodic at $\tau$. If $\ell$ is odd, then 
$\bmu=\begin{bmatrix}
\be_u-\frac{2}{a}\uno_a \\ s\big(\be_v-\frac{2}{b}\uno_b\big).
\end{bmatrix}$ has two nonzero entries if and only if $(a,b)\in\{(1,2),(2,2)\}$. From this, 2a and 2b are immediate. 

\noindent \textbf{Case 2.} Next, suppose $\sin( \tau\sqrt{ab})\neq 0$. 
Suppose we write (\ref{Eq:pf}) as $\bmu:=\begin{bmatrix}
\bw_1 \\ \bw_2\end{bmatrix}$ where $\bw_1\in\C^a$. We make the following observations for the vector $\bw_1$. First, we have $\bw_1=\be_u$ if and only if 
\begin{equation}
\label{eqkmn1}
s=-\frac{\frac{1}{a}(\cos(\tau\sqrt{ab})-1)}{\ii\sin(\tau\sqrt{ab})/\sqrt{ab}}=-\ii \tan(\tau\sqrt{ab}/2)\sqrt{b/a}.
\end{equation}
If $a\geq 2$, then $\bw_1$ has $a-1$ nonzero entries if and only if 
\begin{equation}
\label{eqkmn2}
s=-\frac{1+\frac{1}{a}(\cos(\tau\sqrt{ab})-1)}{\ii\sin(\tau\sqrt{ab})/\sqrt{ab}}=\bigg(\sqrt{ab}\csc(\tau\sqrt{ab})-\tan(\tau\sqrt{ab}/2)\sqrt{b/a}\bigg)\ii,
\end{equation}
in which case $\bw_1$ has all entries equal to $-1$
except for the $u$th entry which is equal to 0. Moreover, if $a=1$, then $\bw_1$ is a scalar equal to 0 if and only if $s$ satisfies (\ref{eqkmn2}). Next, we make the following observations for the vector $\bw_2$. Note that $\bw_2=s\be_v$ if and only if 
\begin{equation}
\label{eqkmn3}
s=-\frac{\ii\sin(\tau\sqrt{ab})/\sqrt{ab}}{\frac{1}{b}(\cos(\tau\sqrt{ab})-1)}=\ii \cot(\tau\sqrt{ab}/2)\sqrt{b/a}.
\end{equation}
\item If $b\geq 2$, then $\bw_2$ has $b-1$ nonzero entries if and only if 
\begin{equation}
\label{eqkmn4}
s=-\frac{\ii\sin(\tau\sqrt{ab})/\sqrt{ab}}{1+\frac{1}{b}(\cos(\tau\sqrt{ab})-1)}=-\bigg(\frac{1}{\sqrt{ab}\csc(\tau\sqrt{ab})-\tan(\tau\sqrt{ab}/2)\sqrt{a/b}}\bigg)\ii.
\end{equation}
in which case $\bw_2$ has all entries equal to $-s$ 
except for the $v$th entry which is equal to 0. Moreover, if $b=1$, then $\bw_2$ is a scalar equal to 0 if and only if $s$ satisfies (\ref{eqkmn4}).
Additionally, if $s$ is not equal to the values in (\ref{eqkmn1}) and (\ref{eqkmn2}), then all entries of $\bw_1$ are nonzero. Similarly, if $s$ is not equal to the values in (\ref{eqkmn3}) and (\ref{eqkmn4}), then all entries of $\bw_2$ are nonzero. Furthermore, (\ref{eqkmn1}) and (\ref{eqkmn2}) cannot simultaneously hold, (\ref{eqkmn3}) and (\ref{eqkmn4}) cannot simultaneously hold, and (\ref{eqkmn1}) and (\ref{eqkmn3}) cannot simultaneously hold (otherwise $\bmu$ is periodic). One also checks that (\ref{eqkmn2}) and (\ref{eqkmn3}) simultaneously hold if and only if $a=2$. We divide the proof into the following subcases.
\begin{itemize}
\item Let $a=1$ (so that $\bw_1$ has one entry). If $\bw_1^T\be_u\neq 0$, then $\bmu$ has exactly two nonzero entries if and only if $\bw_2$ has exactly one nonzero entry. Equivalently, one of the following conditions holds.
\begin{itemize}
\item $b\geq 1$ and $s$ satisfies (\ref{eqkmn3}), in which case $\bmu=-\be_u+s\be_v$. This establishes 2c.
\item $b=2$ and $s$ satisfies (\ref{eqkmn4}), in which case $\bmu=\left(\cos( \tau\sqrt{2})+\ii \frac{s}{\sqrt{2}}\sin(\tau\sqrt{2})\right)\be_u -s\be_w$, where $B_2=\{v,w\}$. This proves 2d.
\end{itemize}
If $\bw_1^T\be_u=0$ (so that (\ref{eqkmn2}) is satisfied), then $\bmu$ has exactly two nonzero entries if and only if $\bw_2$ has exactly two nonzero entries. Equivalently, one of the following holds.
\begin{itemize}
\item $b=2$ and $\bmu=\left(\ii\frac{1}{\sqrt{2}}\sin(\tau\sqrt{2})+\frac{s}{2}(\cos(\tau\sqrt{2})+1)\right)\be_v+\left(\ii\frac{1}{\sqrt{2}}\sin(\tau\sqrt{2})+\frac{s}{2}(\cos(\tau\sqrt{2})-1)\right)\be_w$ where $B_2=\{v,w\}$. In this case, (\ref{eqkmn2}) simplifies to $s=-\sqrt{2}\cot(\tau\sqrt{2})$, so 2e holds.
\item $b=3$ and $s$ satisfies (\ref{eqkmn4}), in which case 
$\bmu=-s(\be_w+\be_x)$ where $B_2=\{v,w,x\}$. As $s$ satisfies (\ref{eqkmn2}) and (\ref{eqkmn4}) and $(a,b)=(1,3)$, we get $\tau\in\big\{\frac{2\pi(3k+\ell)}{3\sqrt{3}}:\ell=1,2\ \text{and}\ k\in\Z\big\}$. In particular, $s=-\ii$ if $\tau=\frac{2\pi(3k+1)}{3\sqrt{3}}$ and $s=\ii$ otherwise. This yields 2f, in which case $s$ is a phase factor.
\end{itemize}
\item Let $a\geq 2$ with $u,w\in B_1$ such that $u\neq w$. Since $s$ cannot simultaneously satisfy (\ref{eqkmn3}) and (\ref{eqkmn4}), it cannot happen that $\bw_1^T\be_u\neq 0$ and $\bw_1^T\be_w\neq 0$. Similarly, since $s$ also cannot satisfy both (\ref{eqkmn1}) and (\ref{eqkmn2}), it cannot happen that $\bw_1^T\be_u=\bw_1^T\be_w=0$. Now, suppose $\bw_1^T\be_u=0$ and $\bw_1^T\be_w\neq 0$ (equivalently, (\ref{eqkmn2}) holds). Then $\bmu$ has exactly two nonzero entries if and only if $a=2$ and $\bw_2$ has exactly one nonzero entry. If $b\geq 1$ and $s$ satisfies (\ref{eqkmn3}), then $\bmu=-\be_w+s\be_v$ has exactly two nonzero entries. This proves 2g. Now, if $b\geq 1$ and $s$ does not satisfy (\ref{eqkmn3}), then $\bmu$ has exactly two nonzero entries if and only if $b=1$ or $b=2$ and $s$ satisfies (\ref{eqkmn4}). The case $b=1$ is equivalent to 2d. The case $b=2$ cannot happen because $a=b=2$ contradicts one of (\ref{eqkmn2}) and (\ref{eqkmn4}).
The case when $\bw_1^T\be_u\neq 0$ and $\bw_1^T\be_w=0$ may be shown to be equivalent to 2g with the roles of $B_1$ and $B_2$ reversed.
\end{itemize}
Combining these two subcases completes the proof of the second case. This proves the forward implication. The converse is immediate.
\end{proof}

We note that in Theorem \ref{Thm:kmnPSTc1} 1 (d) and (e), depending on the choice of $\tau,$ the entanglement degrees may differ in $\bu$ and $\bmu$. For example, in Theorem \ref{Thm:kmnPSTc1} 1 (c), it is straightforward to show that if $\tau = \pm \frac{\tan^{-1}(\sqrt{2})}{\sqrt{2}}$ then $\bmu$ has maximum entanglement, while the value of $s$ is given by $\frac{(\sqrt{3}-1)^2}{2}$. 

\subsubsection{Laplacian case}

The spectral decomposition of the transition matrix of $K_{a,b}$ relative to $L$ is given by
\begin{equation}
\label{Eq:completebipL}
U_{L}(t)=\frac{1}{n}J+\frac{e^{\ii tn}}{abn}\left[\begin{array}{cc} b^2J&-abJ \\ -abJ&a^2J\end{array} \right]+e^{\ii tb}\left[\begin{array}{cc} I_a-\frac{1}{a}J&\zero \\ \zero&\zero\end{array} \right]+e^{\ii ta}\left[\begin{array}{cc} \zero&\zero \\ \zero&I_b-\frac{1}{b}J\end{array} \right].
\end{equation}

\begin{theorem}
\label{Thm:kmnPST1}
Let $V(K_{a,b})=B_1\cup B_2$, where $|B_1|=a$. Suppose $\bu=\be_u+s\be_v$ for some $s\in\C\backslash\{0\}$ and $\bmu\in\C^n$ is a vector with two nonzero entries. PST occurs from $\bu$ to $\bmu$ in $K_{a,b}$ relative to $L$ at time $\tau$ if and only if one of the following conditions hold.
\begin{enumerate}
\item  $u,v\in B_1$, and one of the conditions below hold.
\begin{enumerate}
\item $\dagger$ $a=2$, $b=1$, $s\in\C\backslash\{0,\pm 1\}$, $e^{\ii\tau}+\frac{1+s}{6}(2+e^{\ii 3\tau}-3e^{\ii \tau})=0$, $\bmu=(s-1)e^{\ii\tau}\be_v+\frac{1+s}{3}(1-e^{\ii3\tau})\be_w$ where $B_2=\{w\}$, and $\tau\in\R\backslash\{\frac{2k\pi}{3}:k\in\Z\}$.
\item $\dagger$ $a=2$, $b=1$,  $s\in\C\backslash\{0,\pm 1\}$ satisfies $se^{\ii\tau}+\frac{1+s}{6}(2+e^{\ii 3\tau}-3e^{\ii \tau})=0$, $\bmu=(1-s)e^{\ii\tau}\be_u+\frac{1+s}{3}(1-e^{\ii3\tau})\be_w$ where $B_2=\{w\}$, and $\tau\in\R\backslash\{\frac{2k\pi}{3}:k\in\Z\}$.
\item $*$ $a=b=2$, $s=1$, $\bmu=-(\be_w+\be_x)$ where $B_2=\{w,x\}$, and $\tau=\frac{\pi}{4}$.
\item $\dagger$
$a=2$, $b\geq 1$, $s\in\C\backslash\{0,-1\}$, $\bmu=e^{\ii b \tau}(\be_u+s\be_v)+\frac{1+s}{2}(1- e^{\ii b \tau})\uno_2$, and $\tau\in\R$ such that $e^{\ii2\tau }=e^{-\ii\tau b }\notin\{1,\pm\frac{s-1}{s+1}\}$.
\item $*$ $a=3$, $b\geq 1$, $s\in\C\backslash\{0,-1,2\}$ satisfies $\operatorname{Re}s=\frac{1}{2}$, $\bmu=\frac{s^2-1}{s-2}\be_v-\frac{s+1}{s-2}\be_w$ where $B_1=\{u,v,w\}$ and $\tau\in\R$ such that 
$e^{\ii3\tau }=e^{-\ii\tau b }=\frac{s-2}{s+1}$.
\item $*$ $a=3$, $b\geq 1$, $s\in\C\backslash\{0,-1,\frac{1}{2}\}$ satisfies $(\operatorname{Re}s-1)^2+(\operatorname{Im}s)^2=1$, $\bmu=\frac{s^2-1}{2s-1}\be_u+\frac{s^2+s}{2s-1}\be_w$ where $B_1=\{u,v,w\}$, and $\tau\in\R$ such that 
$e^{\ii3\tau}=e^{-\ii\tau b }=\frac{1-2s}{s+1}$.
\item 
$*$ $a=4$, $b=4k$ for any odd $k$, $s=1$, $\bmu=\be_w+\be_x$ where $B_1=\{u,v,w,x\}$ and $\tau=\frac{\pi}{4}$.
\end{enumerate}
\item $u\in B_1$, $v\in B_2$ and one of the conditions below hold.
\begin{enumerate}
\item $*$ $a=b=2$, $s\in\C\backslash\{0\}$, $\bmu=\be_w+s\be_x$, and $\tau=\frac{\pi}{2}$.
\item $\dagger$ $a=1$, $b\geq 1$, $s=\frac{b(1-e^{\ii \tau (b+1)})}{(b+1)e^{\ii \tau}-b-e^{\ii \tau (b+1)}}$, $\bmu=\left(e^{\ii b \tau}+\frac{(1+be^{\ii \tau (b+1)}-ne^{\ii b \tau})}{b+1}+\frac{s(1-e^{\ii \tau (b+1)})}{b+1}\right)\be_u+se^{\ii \tau }\be_v$, and $\tau\in\R$ such that $e^{\ii\tau(b+1)}\neq 1$.
\item $*$ $a=1$, $b=2$, $s=\frac{2(e^{\ii 3\tau}-1)}{2+e^{\ii 3\tau}+3e^{\ii \tau}}\neq \frac{1+2e^{\ii 3\tau}}{e^{\ii 3\tau}-1}$ with $e^{\ii 3\tau}\neq 1$, $\bmu=\frac{1}{3}(1+2e^{\ii 3\tau}+s(1-e^{\ii 3\tau}))\be_u-se^{\ii \tau}\be_w$ where $B_2=\{v,w\}$, and $\tau\in\R$.
\item $\dagger$ $a=1$, $b=2$, $s=-\frac{1+2e^{\ii 3 \tau}}{1-e^{\ii 3 \tau}}$ and $e^{\ii 3 \tau}\neq 1$, $\bmu=\left(se^{\ii \tau}+\frac{s}{6}(2+e^{\ii 3 \tau}-3e^{\ii \tau})+\frac{1}{3}(1-e^{\ii 3 \tau})\right)\be_v+\left(\frac{s}{6}(2+e^{\ii 3 \tau}-3e^{\ii \tau})+\frac{1}{3}(1-e^{\ii 3 \tau})\right)\be_w$ where $B_2=\{v,w\}$, and $\tau\in\R$.
\item $*$ $a=2$, $b\geq 3$ is odd, $s=\frac{b}{2}$, $\bmu=\be_w+s\be_v$ where $B_1=\{u,w\}$, at $\tau=\pi$. 
\item $*$ $a=2$, $b=4k$ for any integer $k$, $s=-1$, $\bmu=\be_w-\be_v$ where $B_1=\{u,w\}$, and $\tau=\frac{\pi}{2}$.
\end{enumerate}
\end{enumerate}
Moreover, the above statements all hold if we interchange the roles of $a$ and $b$, as well as $B_1$ and $B_2$.
\end{theorem}

\begin{proof}
First, suppose $u,v\in B_1$ so that $a\geq 2$. Then $s\neq -1$, otherwise $\bu$ is a fixed state. We use (\ref{Eq:completebipL}) to obtain an expression for $\bmu=U_L(\tau)\begin{bmatrix}
\bu_1\\ \bu_2\end{bmatrix}$, where $\bu_1=\be_u+s\be_v\in\C^a$ and $\bu_2=\zero_b$. One checks that
\begin{equation*}
\bmu=\begin{bmatrix}
e^{\ii b \tau}(\be_u+s\be_v)+\frac{1+s}{an}(a+be^{\ii n \tau}-ne^{\ii b \tau})\uno_a\\ 
\frac{1+s}{n}(1-e^{\ii n \tau})\uno_b
\end{bmatrix}:=\begin{bmatrix}\bw_1\\\bw_2\end{bmatrix} 
\end{equation*}
where $\bw_1\in\C^a$. We proceed with two cases.

\noindent \textbf{Case 1.} Let $e^{\ii n \tau}\neq 1$ so that $\frac{1+s}{n}(1-e^{\ii n \tau})\uno_b\neq \zero$. Then $\bmu$ has two nonzero entries if and only if either (i) $b=1$ and $\bw_1$ has one nonzero entry or (ii) $b=2$ and all entries of $\bw_1$ are zero.
\begin{itemize}
\item Suppose condition (i) holds. Then $a\in\{2,3\}$, otherwise $\bw_1$ has at least two nonzero entries. If $a=3$, then $\bw_1$ has one nonzero entry if and only if $s=1$ and $3+e^{\ii 4\tau}+2e^{\ii \tau}=0$, implying that $e^{\ii 4\tau}=e^{\ii \tau}=-1$, a contradiction. If $a=2$, $\bw_1$ has one nonzero entry if and only if $s\neq 1$, and either $e^{\ii\tau}+\frac{1+s}{6}(2+e^{\ii 3\tau}-3e^{\ii \tau})=0$ or $se^{\ii\tau}+\frac{1+s}{6}(2+e^{\ii 3\tau}-3e^{\ii \tau})=0$, respectively giving us $\bmu=(s-1)e^{\ii\tau}\be_v+\frac{1+s}{3}(1-e^{\ii3\tau})\be_w$ and $\bmu=(1-s)e^{\ii\tau}\be_u+\frac{1+s}{3}(1-e^{\ii3\tau})\be_w$. This proves 1a and 1b.
\item Suppose condition (ii) holds. Then $a=2$, otherwise $\bw_1$ has at least one nonzero entry. In this case, $\bmu$ has two nonzero entries if and only if $s=1$ and $\tau=\frac{\pi}{4}$. This establishes 1c.
\end{itemize}

\noindent \textbf{Case 2.} Suppose $e^{\ii n \tau}=1$ so that $\bw_1=e^{\ii b \tau}(\be_u+s\be_v)+\frac{1+s}{a}(1-e^{\ii b \tau})\uno_a$ and $\bw_2=0$. If $e^{\ii b \tau}=1$, then $\bu$ is periodic, and so we may assume that $e^{\ii b \tau}\neq 1$. Then $\bmu$ has two nonzero entries if and only if $\bw_1$ has two nonzero entries, which can only happen if $a\in\{2,3,4\}$. Suppose $a=2$. Then $\bw_1$ has two nonzero entries if and only if both $e^{\ii b \tau}+\frac{1+s}{2}(1-e^{\ii b \tau})$ and $se^{\ii b \tau}+\frac{1+s}{2}(1-e^{\ii b \tau})$ are nonzero. Equivalently, $e^{\ii b \tau}\neq \pm \frac{s+1}{s-1}$. This proves 1d. Next, suppose $a=3$. Then $\bw_1$ has two nonzero entries if and only if exactly one of $e^{\ii b \tau}+\frac{1+s}{3}(1-e^{\ii b \tau})$ or $se^{\ii b \tau}+\frac{1+s}{3}(1-e^{\ii b \tau})$ is equal to 0. The former yields $e^{\ii b \tau}=\frac{s+1}{s-2}$, and so $\operatorname{Re}s=\frac{1}{2}$ and $s\neq 2$, In this case, we obtain $\bmu=\frac{s^2-1}{s-2}\be_v-\frac{s+1}{s-2}\be_w$. This proves 1e. The latter yields $e^{\ii b \tau}=\frac{s+1}{1-2s}$, and so $(\operatorname{Re}s-1)^2+(\operatorname{Im}s)^2=1$ and $s\neq \frac{1}{2}$. In this case, we obtain $\bmu=\frac{s^2+s}{s-2}\be_v-\frac{s^2+s}{1-2s}\be_w$. This establishes 1f. For the last subcase, suppose that $a=4$. Then $\bw_1$ has two nonzero entries if and only if $s=1$ and $e^{\ii b \tau}=-1$, in which case $\bmu=\be_w+\be_x$ where $B_1=\{u,v,w,x\}$. Since $\bmu$ is a real vector that is periodic at $\frac{2\pi}{\operatorname{gcd}(a,b)}$, the minimum PST time must be $\tau=\frac{\pi}{\operatorname{gcd}(a,b)}$ \cite[Lemma 5.1(2)]{godsil2025perfect}. Since $e^{\ii b \tau}=-1$ and $e^{\ii n \tau}=1$, we get $e^{\ii 4\tau}=-1$, and so $\tau=\frac{\pi}{4}$. Thus, $b=4k$ for some odd $k$. This gives us 1g.

Next, suppose $u\in B_1$ and $v\in B_2$. Letting $\bu_1=\be_u\in\C^a$ and $\bu_2=s\be_v$, (\ref{Eq:completebipL}) yields the expression
\begin{equation}
\label{Eq:kmnlapsh}
\bmu=U_L(t)\begin{bmatrix}
\bu_1\\ \bu_2\end{bmatrix}=\begin{bmatrix}
e^{\ii b \tau}\be_u+\left(\frac{1}{an}(a+be^{\ii n \tau}-ne^{\ii b \tau})+\frac{s}{n}(1-e^{\ii n \tau})\right)\uno_a\\ se^{\ii a \tau}\be_v+\left(\frac{s}{bn}(b+ae^{\ii n \tau}-ne^{\ii a \tau})+\frac{1}{n}(1-e^{\ii n \tau})\right)\uno_b
\end{bmatrix}:=\begin{bmatrix}\bw_1\\\bw_2\end{bmatrix}. 
\end{equation}
We again proceed with cases.

\noindent \textbf{Case 1.} Let $e^{\ii n \tau}=1$, so that $\bw_1=e^{\ii b \tau}\be_u+\frac{1}{a}\left(1-e^{\ii b \tau}\right)\uno_a$ and $\bw_2=s\left(e^{\ii a \tau}\be_v+\frac{1}{b}(1-e^{\ii a \tau})\right)\uno_b$. In this case, $e^{\ii b \tau}=1$ if and only if $e^{\ii a \tau}=1$. Also, if $e^{\ii b \tau}=1$, then $\bu=\bmu$, so we may assume that $e^{\ii a \tau},e^{\ii b \tau}\neq 1$.
\begin{itemize}
\item Let $a=1$, so that $\bw_1=\be_u$. Since $b\geq 2$, $\bmu$ has two nonzero entries if and only if $b=2$ and $e^{\ii \tau}+\frac{1}{2}(1-e^{\ii \tau })=0$. This case yields $\tau=\pi$, and so $e^{\ii b \tau}=e^{\ii2\tau}=1$, a contradiction.
\item Let $a=2$ so that $\bw_1=e^{\ii b \tau}\be_u+\left(\frac{1}{2}(1-e^{\ii b \tau})\right)\uno_2$. Since $\frac{1}{2}(1-e^{\ii b \tau})\neq 0$, $\bmu$ has two nonzero entries if and only if either (i) $e^{\ii b \tau}+\frac{1}{2}(1-e^{\ii b \tau})=0$ and $b=1$, or (ii) $e^{\ii b \tau}+\frac{1}{2}(1-e^{\ii b \tau})=0$, $b=2$ and $e^{\ii 2\tau}+\frac{1}{2}(1-e^{\ii 2\tau})=0$. Note that (i) yields $\tau=\pi$, and so $e^{\ii a \tau}=e^{\ii2\tau}=1$, a contradiction. Meanwhile, (ii) yields $\tau=\frac{\pi}{2}$ and $\bmu=\be_w+s\be_x$ where $B_1=\{u,w\}$ and $B_2=\{v,x\}$. This proves 2a.
\end{itemize}

\noindent \textbf{Case 2.} Suppose $e^{\ii n \tau}\neq 1$. We first make observations about $\bw_1$. Note that $\bw_1=e^{\ii b \tau}\be_u$ if and only if
\begin{equation}
\label{eqkmnL1}
s=\frac{ne^{\ii b \tau}-a-be^{\ii n \tau}}{a(1-e^{\ii n \tau})}
\end{equation}
If $a\geq 2$, then $\bw_1$ has $a-1$ nonzero entries if and only if
\begin{equation}
\label{eqkmnL2}
s=-\frac{n(a-1)e^{\ii b \tau}+a+be^{\ii n \tau}}{a(1-e^{\ii n \tau})},
\end{equation}
in which case all entries of $\bw_1$ are equal to $-e^{\ii b \tau}$ 
except for $u$th entry which is equal to zero. Moreover, if $a=1$, then $\bw_1=0$ if and only if $s$ satisfies (\ref{eqkmnL2}).
Similarly, $\bw_2=se^{\ii a \tau}\be_v$ if and only if
\begin{equation}
\label{eqkmnL3}
s=\frac{b(1-e^{\ii n \tau})}{ne^{\ii a \tau}-b-ae^{\ii n \tau}}.
\end{equation}
If $b\geq 2$, then $\bw_2$ has $b-1$ nonzero entries if and only if
\begin{equation}
\label{eqkmnL4}
s=-\frac{b(1-e^{\ii n \tau})}{n(b-1)e^{\ii a \tau}+b+ae^{\ii n \tau}}
\end{equation}
in which case all entries of $\bw_2$ are equal to $-se^{\ii a \tau}$, except for $v$th entry which is equal to zero. Moreover, if $b=1$, then $\bw_2=0$ if and only if $s$ satisfies (\ref{eqkmnL4}).
Additionally, if $s$ is not equal to the values in (\ref{eqkmnL1}) and (\ref{eqkmnL2}), then all entries of $\bw_1$ are nonzero. Similarly, if $s$ is not equal to the values in (\ref{eqkmnL3}) and (\ref{eqkmnL4}), then all entries of $\bw_2$ are nonzero. Furthermore, (\ref{eqkmnL1}) and (\ref{eqkmnL2}) cannot simultaneously hold, (\ref{eqkmnL3}) and (\ref{eqkmnL4}) cannot simultaneously hold, and (\ref{eqkmnL1}) and (\ref{eqkmnL3}) cannot simultaneously hold (otherwise $\bmu$ is periodic). 
We divide the proof into 2 subcases.

\begin{itemize}
\item Let $a=1$ (so that $\bw_1$ has one entry). If $\bw_1^T\be_u\neq 0$ (so that (\ref{eqkmnL2}) does not hold), then $\bmu$ has exactly two nonzero entries if and only if $\bw_2$ has exactly one nonzero entry. Equivalently, either: 
\begin{itemize}
\item $b\geq 1$ and (\ref{eqkmnL3}) holds with $\bmu=\left(e^{\ii b \tau}+\frac{1}{n}(1+be^{\ii n \tau}-ne^{\ii b \tau})+\frac{s}{n}(1-e^{\ii n \tau})\right)\be_u+se^{\ii \tau }\be_v$.
\item $b=2$ and (\ref{eqkmnL4}) holds, with $\bmu=\frac{1}{3}\left(1+2e^{\ii 3 \tau}+s(1-e^{\ii 3 \tau})\right)\be_u -se^{\ii \tau }\be_w$, where $B_2=\{v,w\}$.
\end{itemize}
These two scenarios establish 2b and 2c, respectively. Now, if $\bw_1^T\be_u=0$ (so that (\ref{eqkmnL2}) is satisfied), then $\bmu$ has exactly two nonzero entries if and only if $\bw_2$ has exactly two nonzero entries. If $b=2$, then we get $\bmu=\left(se^{\ii \tau}+\frac{s}{6}(2+e^{\ii 3 \tau}-3e^{\ii \tau})+\frac{1}{3}(1-e^{\ii 3 \tau})\right)\be_v+\left(\frac{s}{6}(2+e^{\ii 3 \tau}-3e^{\ii \tau})+\frac{1}{3}(1-e^{\ii 3 \tau})\right)\be_w$
where $B_2=\{v,w\}$. This proves 2d.
If $b\geq 3$, then $\bw_2$ has two nonzero entries if and only if $b=3$ and $s$ satisfies (\ref{eqkmnL4}). 
Since (\ref{eqkmnL2}) and (\ref{eqkmnL4}) hold and $(a,b)=(1,3)$, we get $s=-\frac{1+3e^{\ii \tau 4}}{1-e^{\ii \tau 4}}=-\frac{3(1-e^{\ii \tau 4})}{8e^{\ii \tau}+3+e^{\ii \tau 4}}$. This yields $2e^{\ii 3 \tau}+3e^{\ii \tau 4}+1=0$, which has no solution. Thus, this case cannot happen.

\item Let $a\geq 2$ with $u,w\in  B_1$ such that $u\neq w$. Similar to the proof of Theorem \ref{Thm:kmnPSTc1}, it cannot happen that $\bw_1^T\be_u\neq 0$ and $\bw_1^T\be_w\neq 0$, and similarly, it cannot happen that $\bw_1^T\be_u=\bw_1^T\be_w=0$. Now, suppose $\bw_1^T\be_u=0$ and $\bw_1^T\be_w\neq 0$ (equivalently, (\ref{eqkmnL2}) holds). Then $\bmu$ has exactly two nonzero entries if and only if $a=2$ and $\bw_2$ has exactly one nonzero entry. If $b\geq 1$ and $s$ satisfies (\ref{eqkmnL3}), then 
\begin{equation*}
   s=\frac{b(1-e^{\ii \tau (b+2)})}{(b+2)e^{\ii 2\tau}-b-2e^{\ii \tau (b+2)}}=-\frac{2+be^{\ii \tau (b+2)}+(b+2)e^{\ii b \tau}}{2(1-e^{\ii \tau (b+2)})},    
\end{equation*}
which holds if and only if $(e^{\ii 2\tau}+1)(e^{\ii b\tau}+1)=0$. If $e^{\ii b\tau}=-1$, then $e^{\ii 2\tau}=1$, and so $\tau=\pi$, $s=\frac{b}{2}$, $b\geq 3$ is odd and $\bmu=\be_w+\frac{b}{2}\be_v$. This proves 2e. On the other hand, if $e^{\ii 2\tau}=-1$, then $e^{\ii b\tau}=1$ because $e^{\ii n\tau}\neq 1$, and so $\tau=\frac{\pi}{2}$, $s=-1$, $b\equiv 0$ (mod 4) and $\bmu=-(\be_w-\be_v)$. Thus, 2f holds.
Now, if $b\geq 1$ and $s$ does not satisfy (\ref{eqkmnL3}), then $\bmu$ has exactly two nonzero entries if and only if $b=1$ or $b=2$ and $s$ satisfies (\ref{eqkmnL4}). The case $b=1$ is equivalent to 2c. The case $b=2$ cannot happen because $a=b=2$ contradicts one of (\ref{eqkmnL2}) and (\ref{eqkmnL4}). The case when $\bw_1^T\be_u\neq 0$ and $\bw_1^T\be_w=0$ may be shown to be equivalent to 2e and 2f with the roles of $B_1$ and $B_2$ reversed.
\end{itemize}
Combining these two subcases completes the proof of the second case. This proves the forward implication. The converse is immediate.
Combining all cases above yields the desired conclusion.
\end{proof}

In Theorem  \ref{Thm:kmnPST1} 1a, b, d,  there are choices of $\tau$ and $s$ so that the degrees of entanglement in $\bu$ and $\bmu$ are different. In  Theorem  \ref{Thm:kmnPST1} 1e) one can show that $s-1=-\overline{s}$ so that the degree of entanglement is preserved. In Theorem  \ref{Thm:kmnPST1} 1f), we have $\frac{s^2+s}{s^2-1}=\frac{s}{s-1}.$ From hypothesis, $|s-1|=1,$ and so $|\frac{s}{s-1}|=|s|$ so that the degree of entanglement is preserved. 
In Theorem \ref{Thm:kmnPST1} 2c, it is straightforward to show that $\frac{1}{3}((1+2e^{\ii 3\tau})+s(1-e^{\ii 3\tau}))$ has modulus one, since it is equal to $\frac{\overline{z}}{z}$ where $s=2+e^{\ii 3\tau} +3e^{\ii \tau}$. In particular, $\bu$ and $\bmu$ exhibit the same degrees of entanglement. In Theorem \ref{Thm:kmnPST1} 2d, there are choices of $\tau$ such that the moduli of $\frac{se^{\ii \tau}+\frac{s}{6}(2+e^{\ii 3 \tau}-3e^{\ii \tau})+\frac{1}{3}(1-e^{\ii 3 \tau})}{\frac{s}{6}(2+e^{\ii 3 \tau}-3e^{\ii \tau})+\frac{1}{3}(1-e^{\ii 3 \tau})} $ and $\frac{\frac{s}{6}(2+e^{\ii 3 \tau}-3e^{\ii \tau})+\frac{1}{3}(1-e^{\ii 3 \tau})} 
{se^{\ii \tau}+\frac{s}{6}(2+e^{\ii 3 \tau}-3e^{\ii \tau})+\frac{1}{3}(1-e^{\ii 3 \tau})}$ are different from that of $s$.

\subsection{Stars and generalizations}

The following examples of the star graph and generalizations allow for the advantage that one can specify the degree of entanglement (the magnitude of $s$) or specify the readout time, either of which may be of use depending on the situation. On one hand, the idea of quantum state transfer and maximal entanglement between two distant qubits is being explored \cite{ghosh2025quantum, aiyejina2023perfect}. On the other hand, optimal solutions for a speed limit for quantum state transfer, i.e.,  the minimal possible evolution time are of interest \cite{yung2006quantum, kay2010perfect}. Thus both research avenues are captured in our examples. 

\begin{ex}\label{eg:star}$*$ \rm{
Let $M=A$ and consider the star $K_{1,n}$ with central vertex $n+1$, where $n\geq 2$. Fix $\tau>0$. For $1\leq j\leq n$, consider $\bu=\be_j+s\be_{n+1}$ and $\bmu=\be_j-s\be_{n+1}$ where $s:=s(n,\tau)=-\ii \frac{1}{\sqrt{n}}\tan \left(\tau \sqrt{n}/2 \right).$
Replacing $s$ with $\frac{1}{s}$ in Theorem \ref{Thm:kmnPSTc1} 2c, we get PST from $\bu$ to $\bmu$ in $K_{1,n}$ for appropriate choices of $\tau\in\R$. By allowing $s$ to be complex, this example shows that PST can occur from $\bu$ to $\bmu$, unlike the case when $s$ is real. Indeed, if $s$ was chosen to be real and $\bv$ is the Perron eigenvector of $A$, then $\bu$ and $\bmu$ are not strongly cospectral because $\bv^T\bu\neq \pm \bv^T\bmu$. Hence, there is no PST between $\bu$ and $\bmu$ by Observation \ref{Thm:charPST}. Moreover, if $s$ is real, then $s$ cannot have arbitrarily large modulus whenever $\bu$ and $\bmu$ are strongly cospectral (see \cite[Corollary 9.3.3]{monterde2025quantum} for bounds on $s$ whenever $s$ is real, and $\bu$ and $\bmu$ are strongly cospectral $s$-pair states). However, in this particular example, we have $s=i \delta$, where $\delta \rightarrow\infty$ as $\tau\rightarrow\frac{k\pi}{\sqrt{n}}^+$ (resp., $\delta \rightarrow-\infty$ as $\tau\rightarrow\frac{k\pi}{\sqrt{n}}^-$) for any odd integer $k$. Thus, if we allow $s$ to be complex, then $s$ may have arbitrarily large modulus whenever $\bu$ and $\bmu$ are strongly cospectral, unlike the case when $s$ is real.

Alternatively, we may start with a specified purely imaginary value of $s$, then select the readout time $\tau$ so as to achieve PST from $\be_j+s\be_{n+1}$ 
to $\be_j-s\be_{n+1}.$ To this end, let $s=-\ii \delta$ where $\delta \in \R \backslash \{0\}.$ Choose $\tau$ such that $\tan\left(\tau\sqrt{n}/2\right) =\sqrt{n}\delta,$ i.e. $\tau=\frac{2}{\sqrt{n}}\arctan(\sqrt{n}\delta).$  For large values of $x,$ we may write $\arctan(x) = \frac{\pi}{2} +\sum_{j=1}^\infty\frac{(-1)^j}{(2j-1)x^{2j-1}}. $
We then choose 
\begin{equation}\label{eq:t0}
 \tau=  \frac{\pi}{\sqrt{n}} +2\sum_{j=1}^\infty\frac{(-1)^j}{(2j-1)\delta^{2j-1} n^j}. 
 \end{equation}
This series converges provided $n > \frac{1}{\delta^2}$. 
So for any real $\delta \ne 0$ and $n > \frac{1}{\delta^2},$ $K_{1,n}$ has PST from $\be_j-i\delta \be_{n+1}$ 
to $\be_j+i\delta\be_{n+1}$ at time $\tau$ given by \eqref{eq:t0}. This approach allows us to pre-specify the degree of entanglement between vertex states $j$ and $n+1$ in $\bu$, then choose the readout time so as to generate PST to $\bmu$. 
}
\end{ex}

\begin{ex}\label{eg:k2n}$*$\rm{
Here we consider a complete bipartite graph $K_{2,n}$. 
Setting $s=\ii \sqrt{n/2} \cot( t\sqrt{n/2}),$ we get PST
from $s\be_j+\be_{n+1}$ to $s\be_j-\be_{n+2}$ for each $j=1, \ldots, n$ by Theorem \ref{Thm:kmnPSTc1} 2g. Hence, if we choose the readout time $t,$ then we may choose $s$ to generate PST from $s \be_j+\be_{n+1}$ to $s \be_j-\be_{n+2}.$ Alternatively, 
suppose $\delta\in\R$ with $n> 2\delta^2>0$ and set 
$\tau:= 
 \frac{\pi}{\sqrt{2n}} +\sum_{j=1}^\infty\frac{(-1)^j 2^j\delta^{2j-1}}{(2j-1) n^j} $. A technique similar to that in Example \ref{eg:star}, 
we get PST from $\ii\delta \be_j+\be_{n+1}$ to $\ii\delta\be_j-\be_{n+2}$ 
at time $\tau.$ So, we may pre-specify the degrees of entanglement between states $j,n+1$ in their associated $s$-pair state, and $j,n+2$ in their associated $s$-pair state, and produce perfect $(s,r)$-state transfer involving the corresponding $s$-pair states 
at time $\tau.$
}
\end{ex}

\begin{ex}\label{ex:reg}$*$\rm{
Suppose that $X$ is an regular integral graph on $n$ vertices with degree $r$ such that $r^2+4n$ is not a perfect square (for example, we may take $X$ to be a conference graph on $n$ vertices, in which case $r=\frac{n-1}{2}$). Consider the cone $X \vee K_1$ on $X$, which has adjacency matrix $A=\begin{bmatrix}
 \hat A &  \uno \\ \uno^T &0 
\end{bmatrix}$, where $\hat A$ is the adjacency matrix of $X$.  Denote the eigenvalues of $\hat A$ by $r, \lambda_j, j=2, \ldots, n.$ The eigenvalues of $A$ consist of $\lambda_2, \ldots, \lambda_n,$ and $\gamma_1, \gamma_2 = \frac{1}{2}(r \pm \sqrt{r^2+4n}).$ For each $j=2, \ldots, n,$ the eigenprojection matrix for $\lambda_j$ for $A$ is given by $\begin{bmatrix} E_{\lambda_j}&0\\0&0\end{bmatrix},$ where $E_{\lambda_j}$ is the corresponding eigenprojector for $\hat A.$ Observe also that for $k=1,2,$ a unit eigenvector for $\gamma_k$ is given by $$
\sqrt{\frac{n}{n+\gamma_k^2}}\begin{bmatrix} \frac{\gamma_k}{n}\uno \\ 1\end{bmatrix}. 
$$
Using the facts that $\sum_{\lambda \in \sigma(\hat A)\backslash \{r\} } E_{\lambda} =I-\frac{1}{n}J,$  that $e^{2\pi \ii \lambda_j}=1, j=2, \ldots, n,$ and that $ e^{\ii \pi r}=(-1)^r,$ 
it now follows that the transition matrix at $2\pi$ is given by $$U(2 \pi) = 
\begin{bmatrix} 
    I-\frac{1}{n}J + \frac{(-1)^r}{n} \left( \cos(\theta) + \frac{\ii r}{\sqrt{r^2+4n}}sin(\theta)\right) J & \frac{\ii 2(-1)^r}{\sqrt{r^2+4n}}sin(\theta)\uno \\  \frac{\ii 2(-1)^r}{\sqrt{r^2+4n}}sin(\theta)\uno^T & 
    (-1)^r\left( \cos(\theta) - \frac{\ii r}{\sqrt{r^2+4n}}sin(\theta) \right)
\end{bmatrix}, 
$$
where $\theta = \pi \sqrt{r^2+4n}.$ Define the complex number $s$ by $$
s = \frac{1}{2n}\left( -r + \ii \sqrt{r^2+4n} \left(\frac{\cos(\theta) - (-1)^r}{\sin(\theta)}\right) \right). 
$$ An uninteresting computation reveals that  $U(2\pi)(\be_j + s\be_{n+1}) = \be_j + \overline{ s}\be_{n+1}$ for any $j\in V(G)$. Thus we have an example of perfect $(s,r)$-state transfer with $s$-pair states   that are related via complex conjugation. 

We note in passing that if $r^2+4n$ happens to be a perfect square (this happens, for instance, if $X$ is the Petersen graph), then $U(2 \pi)=I$ and $X\vee K_1$ is a periodic graph.}    
\end{ex}


\section{Using fractional revival for perfect $(s,r)$-state transfer}\label{sec:FR}

Let $X$ be a graph and $S\subseteq V(X)$ with $|S|=k$. We say that $X$ admits $k$-\textit{fractional revival} amongst the vertices in $S$ if the transition matrix of $X$ is block diagonal at some $\tau>0$ with one block corresponding to the vertices in $S$. Equivalently, $U(\tau)$ restricted to the entries indexed by the vertices in $S$ is a unitary matrix. In particular, if $k=2$, then we say that the two vertices in $S=\{a,b\}$ admit \textit{fractional revival}, which is equivalent to saying that for some $\alpha,\beta\in\C$ such that $\alpha^2+\beta^2=1$, we have
\begin{equation*}
U(\tau)\be_a=\alpha\be_a+\beta\be_b.
\end{equation*}

In this section, we demonstrate how fractional revival can be utilized to generate examples of perfect $(s,r)$-state transfer. Here, we make use of the Cartesian product of two graphs $X$ and $Y$, which is the graph $X\square Y$ whose adjacency matrix is $A(X)\otimes I+I\otimes A(Y)$. It is known that the transition matrix of $X\square Y$ relative to $A$ or $L$ is given by
\begin{equation*}
U_X(t)\otimes U_Y(t).
\end{equation*}

In \cite{chan2019quantum}, a characterization of fractional revival and generic constructions of graphs with $k$-fractional revival are provided. In particular, one of their constructions \cite[Corollary 3.2]{chan2019quantum} uses the Cartesian product operation. We adopt the idea in this construction to generate $k$-fractional revival more generally from periodic graphs.

\begin{prop}
\label{propCartprod}
Suppose $X$ is a periodic graph at time $\tau$, and $Y$ is not a periodic graph at time $\tau$. Then for each $a\in V(X)$, $X\square Y$ admits $|V(Y)|$-fractional revival amongst the vertices in $\{(a,b):b\in V(Y)\}$.
\end{prop}

\begin{proof}
Since $X$ is a periodic graph at time $\tau$, we have $U(\tau)=\gamma I$. Thus,
$U_X(\tau)\otimes U_Y(\tau)=\gamma(I\otimes U_Y(\tau))$. Since $Y$ is not a periodic graph at time $\tau$, the matrix $I\otimes U_Y(\tau)$ is block diagonal but not diagonal.
\end{proof}
Note that Proposition \ref{propCartprod}
holds for regular graphs $X$ and $Y$ such that $A(X)$ all have integer eigenvalues and $A(Y)$ has a non-integer eigenvalue, in which case it applies to both $A$ and $L$.

To this end, assume that $Y$ has PST between vertices $0$ and $1$ at time $\tau$ with phase factor $\gamma$. That is, $U_{Y}(\tau)\be_{0}=\gamma\be_{1}$. Observe that
\begin{equation}
\label{Eq:Cart}
\begin{split}
U_{X\square Y}(\tau)(\be_{(0,0)}+s\be_{(1,0)})&=\big(U_{X}(\tau)\otimes U_{Y}(\tau)\big)
\big((\be_{0}\otimes \be_{0})+s(\be_{1}\otimes\be_{0})\big)\\
&=U_{X}(\tau)\be_{0}\otimes U_{Y}(\tau)\be_{0}+sU_{X}(\tau)\be_{1}\otimes U_{Y}(\tau)\be_{0}\\
&=\gamma \big(U_{X}(\tau)\be_{0}\otimes\be_{1}+sU_{X}(\tau)\be_{1}\otimes \be_{1}\big)\\
&=\gamma \big(U_{X}(\tau)\be_{0}+sU_{X}(\tau)\be_{1}\big)\otimes\be_{1}.
\end{split}
\end{equation}
Using (\ref{Eq:Cart}), we generate examples of perfect $(s,r)$-state transfer from graphs with $k$-fractional revival.

\begin{ex}
\label{4.2}
{\rm{
 $*$ Suppose $X$ has $3$-fractional revival amongst vertices $\{0,1,2\}$ at time $\tau$. We may write
$$U_{X}(\tau)\be_{0}=\begin{bmatrix}\alpha \\ \beta \\ \omega\end{bmatrix}\quad \text{and}\quad U_{X}(\tau)\be_{1}=\begin{bmatrix}\beta \\ \eta \\ \zeta\end{bmatrix}$$
for some $\alpha,\beta,\omega,\eta,\zeta\in\C$, and so
$$U_{X}(\tau)\be_{0}+sU_{X}(\tau)\be_{1}=\begin{bmatrix}\alpha+s\beta\\ \beta+s\eta \\ \omega+s\zeta\end{bmatrix}.$$
Thus, if $s=-\frac{\omega}{\zeta}$, then $U_{X}(\tau)\be_{0}+sU_{X}(\tau)\be_{1}=\begin{bmatrix}\alpha-\frac{\beta\omega}{\zeta}\\ \beta-\frac{\eta\omega}{\zeta} \\ 0\end{bmatrix}$. Invoking (\ref{Eq:Cart}), we get
\begin{center}
$U_{X\square Y}(\tau)(\be_{(0,0)}+s\be_{(1,0)})
=\gamma \begin{bmatrix}\alpha-\frac{\beta\omega}{\zeta}\\ \beta-\frac{\eta\omega}{\zeta} \\ 0\end{bmatrix}\otimes \be_{1}=\gamma\bigg(\big(\alpha-\frac{\beta\omega}{\zeta}\big)\be_{(0,1)}+\big(\beta-\frac{\eta\omega}{\zeta}\big)\be_{(1,1)}\bigg).$
\end{center}
Therefore, PST occurs from $\be_{(0,0)}+s\be_{(1,0)}$ to $\big(\alpha-\frac{\beta\omega}{\zeta}\big)\be_{(0,1)}+\big(\beta-\frac{\eta\omega}{\zeta}\big)\be_{(1,1)}$ at time $\tau$. Similarly, if $s=-\frac{\beta}{\eta}$ and $s=-\frac{\alpha}{\beta}$, then PST occurs from $\be_{(0,0)}+s\be_{(1,0)}$ to the vectors $\big(\alpha-\frac{\beta^2}{\eta}\big)\be_{(0,1)}+\big(\omega-\frac{\beta\zeta}{\eta}\big)\be_{(2,1)}$ and $\big(\beta-\frac{\alpha\eta}{\beta}\big)\be_{(1,1)}+\big(\omega-\frac{\alpha\zeta}{\beta}\big)\be_{(2,1)}$, respectively, at time $\tau$.}}

{\rm{
Next we examine the relationship between $r$ and $s$. Note that $r$ can be written as $r=\frac{\beta + \eta s}{\alpha + \beta s} = \frac{\beta \xi - \eta \omega}{\alpha \xi- \beta \omega}.$ Observe that $\beta \xi - \eta \omega$ is the determinant of the submatrix of $U(\tau)$ on rows $1,2$ and columns $2,3$, while $\alpha \xi- \beta \omega$ is the determinant of the submatrix of $U(\tau)$ on rows $1,2$ and columns $1,3.$ Recalling that $U(\tau)$ is complex, symmetric and unitary, and applying the cofactor formula for the inverse, we find that $\beta \xi - \eta \omega = \overline{\omega} \det(U(\tau))$ while  $\alpha \xi- \beta \omega=- \overline{\xi} \det(U(\tau)).$ Hence 
\begin{equation*}
r=\frac{\overline{\omega} \det(U(\tau))}{- \overline{\xi} \det(U(\tau))} = \overline{s},
\end{equation*}
implying that that the degrees of entanglement in this case is preserved.}}
\end{ex}

\begin{ex} {\rm{
$*$ Suppose $X$ has $(\alpha,\beta)$-fractional revival between vertices $0$ and $1$ at time $\tau$ for some $\alpha,\beta\in\C$ such that $\alpha^2+\beta^2=1$ and $\beta\neq 0$. Then $U_{X}(\tau)\be_{0}=\alpha\be_0+\beta\be_1$ and $U_{X}(\tau)\be_{1}=\beta\be_0+ \delta\be_1,$
where $\delta=-\frac{\overline{\alpha}\beta}{\overline{\beta}}$. The same argument as in Example \ref{4.2} yields
\begin{equation*}
\begin{split}
U_{X\square Y}(\tau)(\be_{(0,0)}+s\be_{(1,0)})
=\gamma \begin{bmatrix}\alpha+s\beta\\ \beta+s\delta \\ \textbf{0}\end{bmatrix}\otimes \be_{1}=\gamma\bigg(\big(\alpha+s\beta\big)\be_{(0,1)}+\big(\beta+s\delta\big)\be_{(1,1)}\bigg).
\end{split}
\end{equation*}
Thus, if $\beta=\frac{1-\alpha}{s}$ and $r=\beta+s\delta$, then PST occurs from $\be_{(0,0)}+s\be_{(1,0)}$ to $\be_{(0,1)}+r\be_{(1,1)}$ at time $\tau$. Since $r=\beta + s\delta =\beta -\frac{\overline{\alpha} \beta s}{\overline{\beta}}$ and $\beta s =1-\alpha,$ it follows that $r=\frac{\beta \overline{\beta} + \alpha \overline{\alpha}-\overline{\alpha}}{\overline{\beta}}= \frac{1-\overline{\alpha}}{\overline{\beta}}.$ As $s=\frac{1-{\alpha}}{{\beta}},$ we see that $r=\overline{s}.$ }}
\end{ex}


\begin{ex} {\rm{
$\dagger$ 
As a special case of the analysis above for fractional revival, suppose that $Y$ has PST from $0$ to $1$ at time $\tau$ with phase factor $\gamma$, and that $s\in \C\setminus\{0\}.$ From (\ref{Eq:Cart}), we find that $$U_{K_2\square Y}(\tau)(\be_{(0,0)}+s\be_{(1,0)}) = \gamma(  (\cos \tau + \ii s\sin \tau)\be_{(0,1)}+(\ii\sin \tau + s\cos \tau)\be_{(1,1)}),
$$ i.e. there is PST from $\be_{(0,0)}+s\be_{(1,0)}$ to $(\cos \tau + \ii s\sin \tau)\be_{(0,1)}+(\ii \sin \tau + s\cos \tau)\be_{(1,1)}$ at time $\tau.$ Set $s=s_1+ \ii s_2$ with $s_1, s_2 \in \R$. Then $|\cos \tau + \ii s\sin \tau|^2=\cos^2\tau + |s|^2 \sin^2 \tau -s_2 \sin 2\tau$ and 
$|\ii\sin \tau + s\cos \tau|^2=\sin^2 \tau +|s|^2\cos^2 \tau+s_2\sin 2\tau$. 
In particular there is maximum entanglement in the image state at time $\tau$ if and only if $\cos^2\tau + |s|^2 \sin^2 \tau -s_2 \sin 2\tau = \sin^2 \tau +|s|^2\cos^2 \tau+s_2\sin 2\tau,$ i.e. if and only if $ (1-|s|^2)\cos 2\tau =2s_2 \sin 2\tau.$ This last equality holds in precisely one of three cases: i) $|s|=1$ and either $s\in \R$ or $\tau$ is an integer multiple of $\frac{\pi}{2}$; ii) $s_2\ne 0$ and $\tan 2\tau = \frac{1-|s|^2}{s_2}$; iii) $s_2=0, |s|\ne 1$ and $\tau$ is an odd multiple of $\frac{\pi}{4}.$ }}

{\rm{
As an example we may take $Y$ to be the weighted $d$-cube with all edge weights equal to $2$ and vertices $0, 1$ as antipodal vertices in $Y$. Then for any $s\in \R\setminus \{-1,1\}$ there is PST from 
$\frac{1}{\sqrt{1+s^2}}(\be_{(0,0)}+s\be_{(1,0)})$ to 
$ \frac{1}{\sqrt{2(1+s^2)}}(1+ \ii s) \be_{(0,1)}+  \frac{1}{\sqrt{2(1+s^2)}}(\ii+s) \be_{(1,1)} $ at time $\frac{\pi}{4}.$ Observe that in the image state,  $ \be_{(0,1)}$ and $ \be_{(1,1)}$ are maximally entangled, regardless of which $s \in \C\setminus\{0\}$ is chosen.}}
\end{ex}

\section{Optimizing $(s,r)$-state transfer}\label{sec:optim}

It is well-known that PST between vertex states is rare---indeed,  for any   $k\in \mathbb N$ there are
only finitely many unweighted graphs with maximum valency $k$ on which perfect
state transfer occurs \cite{godsil2012can}. In the study of edge state transfer (where there is an edge between the vertices of the $s$-pair states) and plus state transfer (where the coefficients of the $s$-pair states are $+1$), the number of graphs with  Laplacian pair-state PST and unsigned Laplacian plus
state PST, although significantly higher than the number of graphs with  adjacency vertex
state PST, are still rather low \cite[Table 3]{Chen2020PairST}. Given the relatively low frequency of PST in these settings, optimization theory provides a promising route forward. In much of the experimental work, perfect fidelity (of exactly 1) is arguably not possible, and there is a growing body of recent work on the topic of high-fidelity state transfer \cite{chen2016asymptotically, coutinho2025peak, lippner2024quantifying, lippner2025strong, vinet2024quantum}. We therefore consider in this section the situation   when PST is not possible---the question becomes, can we maximize it? This is different
than PGST in that we do not have (or want) an infinite sequence of times where the fidelity gets asymptotically close to 1. 

Recently, the notion of peak state transfer was introduced \cite{coutinho2025peak}, whereby one is interested in the highest fidelity possible between two given vertices in a continuous-time quantum walk. For a matrix $M$ (resp., vector $\bu$) with complex entries, let $|M|$ (resp., $|\bu|$) be the matrix (resp., vector) obtained from $M$ (resp., $\bu$) by taking the absolute value of all its entries. Given an $n\times n$ symmetric matrix $H$ with spectral decomposition $H=\sum_{j=0}^d\lambda_jE_j$, its \textit{bounding matrix} is defined as 
\begin{equation*}
B(H)=\sum_{j=0}^d|E_j|.
\end{equation*}
Note that for any pair of vertices $u$ and $v$ in $X$, we have $B(H)_{u,v}\geq |U(t)_{u,v}|$ for all $t\in\R$. If there exists a time $\tau$ such that the inequality is equality, then there is \textit{peak state transfer} between $u$ and $v$  at time $\tau$.

It is natural to examine the spectral characterization of peak state transfer found in \cite{coutinho2019perfect} to produce analogous results in the setting of s-pair state transfer. Let 
$H =\sum_{j=0}^d\lambda_j E_j$ be the spectral decomposition of $H$. Then for any 
pure states $\bu$ and $\bmu$, we have 
\begin{equation*}
\left|\bmu^TU(t)\bu\right|= \left|\bmu^T\sum_{j=0}^de^{\ii t\lambda_j}E_j\bu\right|\leq \sum_{j=0}^d\left|e^{\ii t\lambda_j}\bmu^TE_j\bu\right|.
\end{equation*}
Since each $e^{\ii t\lambda_j}$ lies on the unit circle, we have 
\begin{equation*}
\left|\bmu^TU(t)\bu\right|\leq \sum_{j=0}^d\left|\bmu^TE_j\bu\right|\leq|\bmu|^T\left(\sum_{j=0}^d\left|E_j\right|\right)|\bu|=|\bmu|^TB(H)|\bu|
\end{equation*}
for all $t$. If there is a time $\tau>0$ such that  $|\bmu^TU(\tau)\bu|=|\bmu|^TB(H)|\bu|$ for pure states $\bmu\neq \bu$, we say that there is \emph{peak $(s,r)$-state transfer} from $\bmu$ to $\bu$. Unfortunately, although we can rotate the states so that at least one entry is real, we cannot always rotate in such a way to force all entries of both $\bu$ and $\bmu$ to be real, thus in the general complex-valued setting, results strictly analogous to \cite{coutinho2025peak} remain elusive (although they are attainable in the real-valued setting).

We now present a different method for optimizing state transfer. Below is an algorithm for finding a vector with two nonzero entries that maximizes the fidelity from an $s$-pair state at a given time. 
Consider a graph with adjacency matrix $A$, and suppose that we are given the $s$-pair state $\bu= \frac{1}{\sqrt{1+|s|^2}}(\be_a+s\be_b)$. Fix a time $t>0,$ and let $U(t)=e^{\ii tA}.$ Which vector with two nonzero entries maximizes the fidelity of transfer to $\bu$ at time $t$? This is different than the notion of peak $(s,r)$-state transfer in that for the latter, the input and output $s$-pair states are fixed and one looks for a time for which the fidelity is maximized. In this way, the algorithm below is independent of the notion of peak $s$-pair state transfer discussed above. 

To answer that question, set ${\bf{z}}=U(t)\bu, $ and consider the vector $\bv=v_\alpha \be_\alpha +v_\beta \be_\beta,   $
where $\alpha, \beta$ are vertices and $v_\alpha, v_\beta \in \C$ with $|v_\alpha|^2+ |v_\beta|^2=1. $ Observe that $\bv^*U(t)\bu = \bv^* {\bf{z}} = \overline{v_\alpha} z_\alpha + \overline{v_\beta} z_\beta.  $ It now follows from the Cauchy-Schwarz inequality that $|\bv^*U(t)\bu|^2 \le |z_\alpha|^2+ |z_\beta|^2 ,$ with equality if and only if up to a phase factor, we have $\bv=\frac{1}{\sqrt{|z_\alpha|^2+ |z_\beta|^2}}  (z_\alpha \be_\alpha +z_\beta \be_\beta).$

The observations above yield the following strategy for maximizing the fidelity of state transfer from $\bu$ to a vector with two nonzero entries at a given time $t$.

\begin{algorithm}
\caption{An algorithm to maximize fidelity}\label{alg}
\begin{algorithmic}
\State 1. Find ${\bf{z}}=U(t)\bu.$
\State 2. Identify distinct indices $\alpha, \beta$ so that $|z_\alpha|^2+ |z_\beta|^2 = \max\{ |z_p|^2+ |z_q|^2|p \ne q\}. $ 
\State 3. Let  $\bv=\frac{1}{\sqrt{|z_\alpha|^2+ |z_\beta|^2}}  (z_\alpha \be_\alpha +z_\beta \be_\beta).$ 
\State \ The maximum fidelity is then given by $|z_\alpha|^2+ |z_\beta|^2.$
\end{algorithmic}
\end{algorithm}
The vector $\bv$ constructed via Algorithm~\ref{alg} maximizes (over all vectors with two nonzero entries) the fidelity to $\bu$ at time $t.$

\begin{ex}{\rm{
Consider $P_5$ and take our $s$-pair state as $\bu= \frac{1}{\sqrt{2}}(\be_2+\be_4).$ At any time $t$ we have $$U(t)\bu = \begin{bmatrix}
    \frac{\ii}{\sqrt{6}} \sin(\sqrt{3}t)\\
    \frac{1}{\sqrt{2}} \cos(\sqrt{3}t)\\
    \ii\sqrt{\frac{2}{3}} \sin(\sqrt{3}t)\\
    \frac{1}{\sqrt{2}} \cos(\sqrt{3}t)\\
    \frac{\ii}{\sqrt{6}} \sin(\sqrt{3}t)
\end{bmatrix}.$$ 
For $t\in [\frac{\pi}{3\sqrt{3}}, \frac{\pi}{2\sqrt{3}}]$ we have $\frac{1}{\sqrt{6}} \sin(\sqrt{3}t) \ge \frac{1}{\sqrt{2}} \cos(\sqrt{3}t) \ge 0,$ so that the entry of 
$U(t)\bu$ of largest modulus is in the third position, while 
the entries of 
$U(t)\bu$ of second-largest modulus are in the first and fifth  positions. It now follows that the vectors $\frac{1}{\sqrt{5}}\be_1 +\frac{2}{\sqrt{5}}\be_3 $ and $\frac{1}{\sqrt{5}}\be_5 +\frac{2}{\sqrt{5}}\be_3 $ both maximize the fidelity of transfer from $\bu,$ with fidelity equal to $\frac{5}{6}\sin^2(\sqrt{3}t).$ In particular, at $t=\frac{\pi}{2\sqrt{3}},$ the fidelity is $\frac{5}{6}=0.8\overline{3}.$ }}
\end{ex}

\begin{ex}\rm{ We consider the star $K_{1,n}$ and continue with the notation of Example \ref{eg:star}. Fix $\tau, s \in \mathbb{R}$ and let $\theta = t\sqrt{n}.$ Let $\bu = \frac{1}{1+s^2}(\be_1+ \ii s\be_{n+1})$ and $\bmu(r) =\frac{1}{1+r^2}(\be_1+ \ii r\be_{n+1})$; we seek to choose $r$ so that the fidelity $|\bmu^*U(\tau)\bu|^2$ is maximized. A computation shows that 
$$|\bmu^*U(\tau)\bu|^2 = \frac{\bigg(1+\frac{ \cos(\theta)-1}{{n}} - \frac{s\sin(\theta)}{\sqrt{n}} - r\big(\frac{\sin(\theta)}{\sqrt{n}} +s\cos(\theta)\big)\bigg)^2 }{(1+s^2)(1+r^2)}. 
$$
If $1+\frac{ \cos(\theta)-1}{{n}} - \frac{s\sin(\theta)}{\sqrt{n}}=0,$ then 
$|\bmu^*U(\tau)\bu|^2$ is increasing in $r,$ with limiting value 
$\frac{(\frac{\sin(\theta)}{\sqrt{n}} +s\cos(\theta))^2 }{(1+s^2)} $ as $r \rightarrow \infty.$ On the other hand, if $1+\frac{ \cos(\theta)-1}{{n}} - \frac{s\sin(\theta)}{\sqrt{n}}\ne 0,$  a calculus exercise shows that $|\bmu^*U(\tau)\bu|^2$  is maximized when $$r = -
\frac{\frac{\sin(\theta)}{\sqrt{n}} +s\cos(\theta)}{1+\frac{ \cos(\theta)-1}{{n}} - \frac{s\sin(\theta)}{\sqrt{n}}}
$$ and that the maximum value for the fidelity is given by
\begin{eqnarray*}
&& 1 -\frac{n-1}{n(1+s^2)}(1-\cos(\theta))\left(s^2(1+\cos(\theta)) +2s\frac{\sin(\theta)}{\sqrt{n}} + \frac{1-\cos(\theta)}{n}\right) = \\
&& \begin{cases} 1 -\frac{n-1}{n(1+s^2)}(1-\cos(\theta)^2) \left(s+ \frac{\sin(\theta)}{\sqrt{n}(1+\cos(\theta))}\right )^2, & \text{if $\cos(\theta)\ne -1$} \\
1-\frac{4(n-1)}{n^2(1+s^2)},  & \text{if $ \cos(\theta)= -1$.} \end{cases} 
\end{eqnarray*}
}
\end{ex}

\section{Sensitivity analysis}\label{sec:sens}

Due to inherent manufacturing and measurement errors, \emph{perfect} state transfer (of any kind--vertex or $s$-pair states) is only perfect theoretically. Small errors, either in setting up the physical system or reading out the spin at exactly the right time, cause the fidelity function to deviate from 1. The shape of the fidelity function $f$ determines the magnitude of these deviations: if $\tau$ is the PST readout time, then $\frac{d f}{d t}\big|_{\tau} =0$ and 
 the second derivative of $f$ is negative on some interval $(\tau -  \epsilon, \tau + \epsilon).$ Thus $f$ is concave down on that interval and the magnitude of the second derivative determines the sensitivity of $f$ to readout time errors in a neighbourhood of $\tau$, i.e. 
 the larger in magnitude the second derivative, the greater the sensitivity.

Precise formulas for derivatives of, as well as bounds for, the probability of (vertex) state transfer with respect to the edge weights and the readout time were derived in \cite{kirkland2015sensitivity,gordon2016bounds}. 
A sensitivity analysis in \cite{godsil2025perfect} produced bounds on the second derivative of the fidelity function in the setting of general real states $\bx$ and $\by$. Here, we apply the latter work to our setting of (complex) $s$-pair states to obtain some results specific to our setting. Our main contribution here is that the second derivative is computable from information from the graph, which is a distinct advantage over the more general analysis of \cite{godsil2025perfect}. 

Suppose that $\bx, \by \in \C^n,$ and that for some symmetric Hamiltonian $M$ and $\tau>1,$ we have $e^{\ii  \tau M}\bx = \gamma \by$ for some unit $\gamma \in \C.$ Set $f(t)=|\by^*e^{\ii  tM}\bx|^2$ which is the fidelity of transfer from $\bx$ to $\by$. As in \cite{godsil2025perfect} (the proof of which  follows that in  \cite{kirkland2015sensitivity}) we find that 
\begin{equation*}
\frac{d^k f}{d t^k}\bigg|_{\tau} = \begin{cases} 
0, & \text{if $k$ is odd} \\
(-1)^{\frac{k \mod 4}{2}}\sum_{j=0}^k(-1)^j {k \choose j } \by^*M^j\by\by^*M^{k-j}\by,  & \text{if $k$ is even.} \end{cases} 
\end{equation*}
In particular, $\frac{d f}{d t}\big|_{\tau} =0$ and $\frac{d^2 f}{d t^2}\big|_{\tau} = -2(\by^*M^2\by - (\by^*M\by)^2).$ Observe that if $\by$ is an $s$-pair state, then the second derivative of the fidelity can be computed by from a $2\times 2$ principal submatrix of the Hamiltonian $M$, and a $2 \times 2$ principal submatrix of $M^2$. 

For concreteness, let $\by=\frac{1}{\sqrt{1+|r|^2}}(\be_\alpha +r\be_\beta)$ for indices $\alpha, \beta$ and $r \in \C\backslash\{0\}$. If we are using the adjacency matrix $A$ as the Hamiltonian, then $\by^*A^2\by - (\by^*A\by)^2$ is readily determined to be given by 
\begin{equation*}
\frac{1}{1+|r|^2}\left(d_\alpha + |r|^2 d_\beta + 2c Re(r)- \frac{4 a_{\alpha, \beta}(Re(r))^2 }{1+|r|^2}\right), 
\end{equation*}
where $d_\alpha, d_\beta$ are the degrees of $\alpha, \beta$ respectively, and $c$ is the number of common neighbours for  $\alpha$ and $ \beta.$ Similarly, if we are using the Laplacian matrix $L$ as the Hamiltonian, we find that $\by^*L^2\by - (\by^*L\by)^2$ equals $$
\frac{1}{1+|r|^2}\left( d_\alpha^2+d_\alpha +|r|^2(d_\beta^2+d_\beta)+2cRe(r)-2(d_\alpha+d_\beta)Re(r)
-\frac{(d_\alpha +|r|^2d_\beta -2a_{\alpha,\beta}Re(r))^2}{1+|r|^2}\right).$$

\begin{ex}\label{eg:starsense} {\rm{
Let $\bu = \be_j+s\be_{n+1}$ and $ \bmu=\be_j-s\be_{n+1}$ as in Example \ref{eg:star} with $1\leq j\leq n,$ and $(-r=) s=-\frac{\ii}{\sqrt{n}}\tan\left(\frac{\tau \sqrt{n}}{2}\right)$ for some $\tau>0.$ The degree of vertex $j$ is 1, that of vertex $n+1$ is $n,$ and they have no common neighbours. 
Using those facts we find that $$\frac{d^2 f}{d t^2}\bigg|_{\tau} =  - 
\frac{2}{1+\frac{1}{n}\tan^2\left(\frac{\tau \sqrt{n}}{2}\right)}\left(
1+\tan^2\left(\frac{\tau \sqrt{n}}{2} \right)\right) = -
\frac{2n}{(n-1)\cos^2\left(\frac{\tau \sqrt{n}}{2}\right)+1}.$$ In particular,     $\frac{d^2 f}{d t^2}\big|_{\tau} \in [-2n,-2].$  }}
\end{ex}

\begin{ex}\label{eg:k2nsense} \rm{Let $\bu=\be_{n+1}+s\be_j$ and $\bmu=-\be_{n+2}+s\be_j$ as in Example \ref{eg:k2n} with $1\leq j\leq n,$ and 
$s=-\ii \sqrt{\frac{n}{2}} \cot\left( \sqrt{\frac{n}{2}} \tau \right)$ for some $t>0.$ Vertex $j$ has degree $2$, vertex $n+2$ has degree $n$, and $s$ is purely imaginary. It now follows that 
$$\frac{d^2 f}{d t^2} = -\frac{4n}{(n-2)\cos^2\left(\tau \sqrt{\frac{n}{2}}\right)  + 2} \in [-2n, -4]. 
$$

}
\end{ex}

\begin{ex}\label{ex:sensitivity_close}\rm{
Recall the  construction in section 3 where the graph $Y$ has PST from 0 to 1 at time $\tau$, and we found that for $K_2 \square Y $ and any $s \in \mathcal{C}\backslash \{0\}$ there is PST from $\be_{(0,0)} + s \be_{(1,0)}$ to  $  (\cos \tau +i s \sin \tau)\be_{(0,1)} + (i \sin \tau + s \cos \tau)\be_{(1,1)}$ at time  $\tau.$ Suppose that vertex 1 has degree $d$ in $Y$, so that vertices $(0,1)$ and $(1,1)$ have degree $d+1$ in  $K_2 \square Y $, and no common neighbours. Our value of $r$ can be taken to be $\frac{\ii \sin \tau + s \cos \tau}{\cos \tau +\ii s \sin \tau}.$ It follows that for $K_2 \square Y $, the second derivative of the fidelity at $\tau$ is given by $-2d-2+8\left(\frac{Re(r)}{1+|r|^2}\right)^2.$ Writing $s$ as $s_1 + is_2, s_1, s_2 \in \mathbb{R},$ a computation shows that $\frac{Re(r)}{1+|r|^2} = \frac{s_1}{1+s_1^2+s_2^2}.$ Thus the second derivative of the fidelity in 
$K_2 \square Y$ is $-2d-2+ 8 \left( \frac{s_1}{1+s_1^2+s_2^2} \right)^2$, which is reasonably close $-2d$, the second derivative of the fidelity in the original graph $Y$. 
}

\end{ex}

\begin{ex}\label{ex:sensitivity_improvement}{\rm{
Let $n$ be divisible by $4$. The cocktail party graph $CP(n)$ on $n$ vertices is regular of degree $n-2$, has integral spectrum, and is known to have PST at time $\frac{\pi}{2}$ between each pair of non-adjacent vertices. The second derivative of the fidelity at time $\frac{\pi}{2}$ between non-adjacent pairs is equal to $-2(n-2)$. 

Next we consider the discussion of Example \ref{ex:reg} for $CP(n)\vee K_1.$ Labelling the vertex of degree $n$ by $n+1$ and the rest by $1, \ldots, n$, we see that for $\theta = \pi \sqrt{n^2+4}$ and 
$s=\frac{1}{2n}\left(-(n-2) + \ii \sqrt{n^2+4}\left(\frac{\cos (\theta)-1}{\sin(\theta)}\right)\right)$, there is perfect $(s,r)$-state transfer from $\frac{1}{1+|s|^2}(\be_j+s\be_{n+1}) $
to $\frac{1}{1+|s|^2}(\be_j+\overline{s}\be_{n+1} )$ at time $2\pi$. The corresponding second derivative of the fidelity (evaluated at time $2\pi$) is given by $$\frac{1}{1+|s|^2} \left(\frac{3n-4}{n} + |s|^2n -\frac{(n-2)^2}{n^2(1+|s|^2)}\right).$$ We have $|s|^2 = \frac{1}{4n^2}\left((n-2)^2+(n^2+4)\left(\frac{(\cos (\theta)-1)^2}{\sin(\theta)^2}\right)\right) = \frac{1}{4n^2}\left( -4n +\frac{2(n^2+4)}{1+\cos(\theta)}\right). $ Observe that for large values of $n, \theta \approx \pi n + \frac{2\pi}{n},$ so that $\cos(\theta) \approx 1,$ and $|s|^2 \approx \frac{1}{4}.  $ It now follows that the second derivative of the fidelity of $s$-pair state transfer is roughly $-2(\frac{n}{5} + \frac{76}{25})$ for large $n.$ This represents an approximately five-fold improvement on the corresponding second derivative for $CP(n).$

}}
    
\end{ex}

\begin{ex}{\rm{ Here we consider the sensitivity in the context of Example \ref{ex:p4'}, and maintain the notation of that example. We have $$\bmu_k = \cos\left(\theta_k+ \frac{\pi}{4}\right)\be_1+ \ii \sin\left(\theta_k+\frac{\pi}{4}\right)\be_4 = \frac{1}{\sqrt{1+\tan\left(\theta_k+\frac{\pi}{4}\right)^2}}(\be_1+ \ii \tan\left(\theta_k+\frac{\pi}{4}\right)\be_4).$$ Hence in the sensitivity expression, we take $r$ as $r=\ii \tan\left(\theta_k+\frac{\pi}{4}\right)$. We have  $Re(r)=0$ while for $P_4$ the degrees of vertices $1$ and $4$ are $1$. It now follows that $\frac{d^2 f}{d t^2}\big|_{\tau_k}=-2,$ independently of the value of $k$. 
}}
    
\end{ex}

\section{Conclusion and future work}\label{sec:concl}
In this work, we expand the horizon beyond quantum state transfer involving vertex states to quantum state transfer involving linear combinations of two vertex states for both the input and output states, called $s$-pair states. While there is recent work on the general pure state setting \cite{godsil2025perfect}, the results in the literature are quite broad. By considering $s$-pair states, we are able to leverage the additional structure to obtain numerous examples.

We have several examples and results of perfect $(s, r)$-state transfer from $\be_a + s\be_b$ to 
 $\be_\alpha + r\be_\beta$ where $|s|$ and $ |r|$ may differ (denoted throughout with a $\dagger$). Since the magnitudes of $r$ and $s$ represent the degrees of entanglement  in the respective states, we see that  the degree of entanglement may not be preserved under PST.  In particular, in the cases when $|s|\neq |r|$, we , if necessary, can take advantage of negative time and/or complex conjugates (as discussed in Section~\ref{sec:back})  to force the input state to have the smaller degree of entanglement. This generates a notion of entanglement boosting (that is, we can if necessary use the symmetry of the quantum state transfer dynamics to arrange it so that the degree of entanglement increases in this setting), which would be potentially physically desirable.

 Our analysis has uncovered numerous examples where $r$ and $s$ are real or purely imaginary, but limited results where $r$ and $s$ are proper complex numbers, leading to the question of the relative rarity of proper complex coefficients. For example, in Proposition~\ref{Cor:knPST}, we are afforded freedom in $s$, yet we are severely limited by the size of the graph. It would be interesting to find other examples of $s$-pair state PST in which both $Re(s)$ and $Im(s)$ are nonzero, and to investigate whether this more general setting offers some advantage (e.g. shorter PST time, better sensitivity properties, or sparser graphs ). 

 Besides introducing the concept of perfect $(s,r)$-state transfer, naturally extending the line of inquiry of vertex- and $s$-pair state transfer, allowing for perfect state transfer between $s$-pair states having differing levels of entanglement,  our contributions herein also include  optimization of state transfer and a sensitivity analysis. In the absence of perfect $(s,r)$-state transfer, our optimization algorithm produces an output vector with two nonzero entries (a bonafide $r$-pair state) that maximizes the fidelity of quantum state transfer from an initial $s$-pair state at a given time. 
 
 In our sensitivity analysis, we considered the case when vertex state transfer is necessary for perfect $(s,r)$-state transfer. 
 Example~\ref{ex:sensitivity_improvement} gave a construction where the second derivative of the fidelity function of the graph having perfect $(s,r)$-state transfer was five times smaller in magnitude than that of the corresponding graph having perfect (vertex) state transfer. In other words, the sensitivity decreased by a factor of roughly 5. However the graphs in that example are quite dense and it would be interesting to identify sparser graphs where by using the construction of Example~\ref{ex:reg} we can decrease the PST sensitivity for $s$-pair states.

\section*{Acknowledgements}
The authors gratefully acknowledge the PIMS-BIRS Teamup initiative for the opportunity for in-person collaboration, as well as  Ada Chan, Sooyeong Kim, and Xiaohong Zhang  for fruitful discussions at the early stages of the project. S.~Plosker would like to thank David Feder for his time in discussing negative time.  S.\ Kirkland was supported by NSERC grant number RGPIN-2025-05547. H.\ Monterde was supported by the University of Manitoba Faculty of Science, Faculty of Graduate Studies and the PIMS-Simons Postdoctoral Fellowship. S.\ Plosker was supported by NSERC Discovery Grant numbers RGPIN-2019-05276 and  RGPIN-2025-05704, the Canada Research Chairs Program grant number 101062, and the  Canada Foundation for Innovation grant number 43948. Some of the work in this paper was completed as part of H.\ Monterde's Ph.D. thesis 
at the University of Manitoba (in particular, Observation \ref{Thm:charPST}, Proposition~\ref{Cor:knPST}, and Theorems~\ref{Thm:kmnPSTc1} and~\ref{Thm:kmnPST1}). 

\bibliographystyle{alpha}
\bibliography{ref}
\end{document}